\documentclass[11pt]{article}
\usepackage[margin=1in]{geometry}
\usepackage[utf8]{inputenc}  % allow utf-8 input
\usepackage[T1]{fontenc}     % use 8-bit T1 fontsreasonable margins
\usepackage{hyperref}       % hyperlinks
\usepackage{url}            % simple URL typesetting
\usepackage{booktabs}       % professional-quality tables
\usepackage{amsfonts}       % blackboard math symbols
\usepackage{nicefrac}       % compact symbols for 1/2, etc.
\usepackage{microtype}      % microtypography
\usepackage{xcolor}         % colors
\usepackage{graphicx}       % Required for inserting images
\usepackage{amsmath, amssymb, amsthm}
\usepackage{algorithm}
\usepackage{algpseudocode}
\usepackage{natbib}

\usepackage{marvosym}

\newtheorem{theorem}{Theorem}
\newtheorem{lemma}{Lemma}

\title{Improved Approximation Algorithms for Chromatic and Pseudometric-Weighted Correlation Clustering}

\author{ 
 \thanks{Authors are listed alphabetically.}
  Chenglin Fan\thanks{ Department of Computer Science and Engineering, Seoul National University, Seoul, 08826, South Korea. 
   \texttt{fanchenglin@snu.ac.kr}  } \and
    Dahoon Lee\thanks{Department of Mathematical Sciences, Seoul National University, Seoul, 08826, South Korea. \texttt{dahoon46@snu.ac.kr}} \and
  Euiwoong Lee\thanks{Computer Science and Engineering Division, University of Michigan, Ann Arbor, MI 48109, USA. \texttt{euiwoong@umich.edu}}
}

\date{}

\begin{document}

\maketitle

\begin{abstract}
 Correlation Clustering (CC) is a foundational problem in unsupervised learning that models binary similarity relations using labeled graphs. While classical CC has been widely studied, many real-world applications involve more nuanced relationships, either multi-class categorical interactions or varying confidence levels in edge labels. To address these, two natural generalizations have been proposed: \emph{Chromatic Correlation Clustering}, which assigns semantic colors to edge labels, and \emph{pseudometric-weighted Correlation Clustering}, which allows edge weights satisfying the triangle inequality. In this paper, we develop improved approximation algorithms for both settings. Our approach leverages LP-based pivoting techniques combined with problem-specific rounding functions. For the pseudometric-weighted correlation clustering problem, we present a tight $\frac{10}{3}$-approximation algorithm, matching the best possible bound achievable within the framework of standard LP relaxation combined with specialized rounding. For the Chromatic Correlation Clustering (CCC) problem, we improve the approximation ratio from the previous best of $2.5$ to $2.15$, and we establish a lower bound of $2.11$ within the same analytical framework, highlighting the near-optimality of our result.
%Our work unify and extend existing LP-rounding techniques, offering new insights into the structural and algorithmic properties of these problems. 
\end{abstract}

\section{Introduction}

Clustering is a fundamental task in unsupervised learning, where the goal is to partition a set of objects into groups based on their pairwise relationships. One prominent problem in this domain is \emph{Correlation Clustering} (CC)~\cite{bansal2004correlation}, which models binary similarity/dissimilarity between items using an edge-labeled graph: similar pairs are marked with a `$+$' label and dissimilar pairs with a `$-$'. The objective is to partition the nodes to minimize disagreements—i.e., cases where the partitioning contradicts the edge labels. Due to its flexibility in not requiring a predefined number of clusters, CC has been widely utilized in various areas such as   detecting communities in networks~\cite{Chen2012}, inferring labels from user interactions~\cite{Agrawal2009,Chakrabarti2008} and  resolving ambiguous entities~\cite{Kalashnikov2008}.

However, classic CC models only binary relationships, which is insufficient for many practical applications. For example, in a social network, edges may represent diverse relationship types such as ``colleague,'' ``classmate,'' or ``family.'' To address this limitation, Bonchi et al.~\cite{bonchi2012chromatic} introduced the \emph{Chromatic Correlation Clustering }(CCC) problem, which generalizes CC to multi-class categorical settings. In CCC, the input is an edge-colored graph where each color represents a different relationship type. The goal is to cluster the nodes and assign a single color to each cluster such that the number of \emph{disagreements}—edges whose color does not match the cluster's assigned color, or edges that should be separated—is minimized. CCC has wide applications in link classification, entity resolution, and clustering in bioinformatics~\cite{bonchi2012chromatic,anava2015chromatic, klodt2021color}.

In parallel, another important generalization of CC is the \emph{weighted correlation clustering problem}, where edges are associated with weights reflecting the reliability or cost of violating a given label. When weights are unrestricted, obtaining a constant-factor approximation is known to be hard (under the Unique Games Conjecture)~\cite{khot2002ugc}. However, when edge weights form a \emph{pseudometric}—i.e., they satisfy the triangle inequality—constant-factor approximations become feasible. This weighted setting more faithfully models scenarios where not all edges are equally trustworthy.

\subsection{Related Works}
The \emph{Correlation Clustering}  problem has been widely studied since its introduction~\cite{ben-dor1999clustering}, and it is known to be APX-hard, leading to efforts to develop approximation algorithms.
Early work by Bansal, Blum, and Chawla introduced a constant-factor approximation algorithm~\cite{bansal2004correlation}. Charikar et al.~\cite{CHARIKAR2005360} improved this to a 4-approximation using linear programming. Ailon, Charikar, and Newman then introduced the \emph{Pivot} algorithm~\cite{ailon2008aggregating}, which achieved a 3-approximation in linear time. Chawla et al.~\cite{chawla2015near} further improved this to 2.06 using more refined LP-rounding techniques.
More recently, researchers have surpassed the 2-approximation barrier. Cohen-Addad, Lee, and Newman~\cite{cohenaddad2022sherali} used the \emph{Sherali-Adams} hierarchy to develop a $(1.994 + \varepsilon)$-approximation, while Cohen-Addad et al.~\cite{cohenaddad2023preclustering} proposed \emph{preclustering}, which improved the approximation to $(1.73 + \varepsilon)$.
The most recent breakthrough by Cao et al.~\cite{cao2024clusterlp} introduced the \emph{cluster LP}, which unifies all known LP relaxations for CC. They show that this can be approximated efficiently using preclustering, achieving a $(1.437 + \varepsilon)$-approximation, the best known guarantee for CC so far. In a more recent work~\cite{cao2024clusterlp}, they introduced a new approach to find a feasible solution for the cluster LP in sublinear time.

Chromatic Correlation Clustering  is an extension of the classical Correlation Clustering  problem, where edge colors represent different types of relationships. Bonchi et al.~\cite{bonchi2012chromatic} introduced CCC with a heuristic lacking guarantees. Anava et al.~\cite{anava2015chromatic} gave a 4-approximation via LP rounding, plus two practical methods: Reduce and Cluster (RC, ratio 11) and Deep Cluster (DC). Klodt et al.~\cite{klodt2021color} showed that Pivot~\cite{ailon2008aggregating} yields a 3-approximation and that RC achieves 5, but Pivot’s color-agnostic design limits performance; they instead proposed Greedy Expansion (GE), a heuristic effective in practice without guarantees.
More recently, Xiu et al.~\cite{XiuHTCH22} developed a 2.5-approximation algorithm for CCC based on a linear programming approach, improving upon the previous best-known ratio. They also introduced a greedy heuristic that achieves strong empirical results.

In modern data analysis, correlation clustering must often be performed under computational constraints such as limited memory or streaming access to data.
Consequently, substantial research has focused on crafting clustering algorithms specifically tailored for dynamic, streaming, online, and distributed settings~\citep{lattanzi2017consistent, fichtenberger2021, jaghargh2019, cohen2019fully, guo2021distributed, cohen2022online, assadi2022sublinear, lattanzi2021parallel, behnezhad2022, behnezhad2023, bateni2023,cohen2024dynamic, braverman2025fully}.

\subsection{Our Results}
\paragraph{Our Contributions.} In this work, we present improved approximation algorithms for both the CCC and pseudometric-weighted CC problems.
\begin{itemize}
\item For the \emph{pseudometric-weighted correlation clustering} problem, we develop a refined LP-based pivoting algorithm that achieves a tight \textbf{$\frac{10}{3}$-approximation}. We further prove that this approximation factor is \emph{optimal} within the standard LP relaxation framework combined advanced rounding functions.

\item For the \emph{Chromatic Correlation Clustering } problem, we enhance the LP-based method through a new analysis that yields a \textbf{$2.15$-approximation}, improving upon the previous best bound of $2.5$ by Xiu et al.~\cite{XiuHTCH22}. We also establish a lower bound of \textbf{$2.11$} within the same analytical framework, underscoring the near-optimality of our approach.
\end{itemize}

Both results are obtained by extending and unifying the triple-based analysis of LP-rounding schemes. Our work improves the theoretical guarantees for two natural and practically motivated generalizations of correlation clustering and contributes new insights into their structural and algorithmic properties.

\paragraph{Technical Overview.}  
Our algorithms for both Chromatic Correlation Clustering (CCC) and pseudometric-weighted Correlation Clustering (CC) build on linear programming (LP) relaxations and a unified triple-based rounding framework~\cite{chawla2015near} . Below, we outline the key technical insights:

\noindent \textbf{Pseudometric-Weighted CC:}  
The upper bound for the approximation factor $10/3$ is derived using the LP-based Pivot algorithm and a more careful rounding function. 
For the lower bound, 
By assuming the existence of an $\alpha$-approximation and analyzing carefully constructed hard instances, the technique derives necessary conditions that any rounding function must satisfy. These conditions expose inherent conflicts, demonstrating that $\alpha$ cannot be arbitrarily small. In particular, the analysis establishes that $\alpha$ must be at least $\frac{10}{3}$. The core idea is to identify instance configurations that induce contradiction between the properties required of the rounding functions, ultimately leading to this lower bound on $\alpha$.

\noindent \textbf{Chromatic Correlation Clustering (CCC):}  
 Building on the LP formulation introduced by Xiu et al.~\cite{XiuHTCH22}, which jointly encodes fractional cluster membership and color assignments.
The decoupling of color assignment from cluster formation, allowing us to preserve color structure without entangling it with clustering decisions. Using a triple-based analysis, we introduce tailored rounding functions—particularly for neutral edges—to better align the rounding behavior with the LP’s structure and avoid overcounting. This careful handling of intra-color, conflicting, and neutral edges reduces the approximation factor from 2.5 to 2.15. Our lower bound analysis builds on the general triple-based framework, augmented with structural insights specific to the LP-CCC algorithm and its associated LP solution. We carefully define the cost and LP contribution of each edge type—particularly neutral edges—and construct adversarial instances that expose limitations of any rounding strategy.
%For the lower bound, the analysis carefully defines the cost structure and LP contributions—particularly for neutral edges—and selecting adversarial instance configurations, the analysis derives conflicting bounds on a key intermediate algorithmic quantity.

%\noindent \textbf{Unified Analysis:}  
%Both results leverage a triple-based inequality framework to bound the expected cost. For each node triple, we analyze how the LP solution and rounding process affect disagreements, then aggregate these bounds globally. This approach unifies the treatment of structural constraints (pseudometric weights or chromatic assignments) and enables tighter approximation guarantees. The improvements highlight the flexibility of LP-based methods when augmented with problem-specific geometric and combinatorial insights.

\paragraph{Paper Organization.} 
The remainder of the paper is structured as follows: 
\textbf{Section~\ref{sec:pre}} introduces the problem formulations and LP relaxations for both pseudometric-weighted and chromatic correlation clustering. 
\textbf{Section~\ref{sec:algo}} presents our approximation algorithms and outlines their design. 
\textbf{Section~\ref{sec:round}} defines the rounding functions used in the LP-based algorithms. 
\textbf{Section~\ref{sec:analysis}} provides a detailed triple-based analysis of the approximation guarantees. 
We conclude with a summary and discussion in \textbf{Section~\ref{sec:con}}.
\section{Preliminaries}
\label{sec:pre}
The \emph{correlation clustering} (CC) problem takes as input a signed undirected graph \( G = (V, E=E^+\uplus E^-) \), where each edge \( e=uv \in E \) is assigned a sign `$+$' or `$-$', described by $e\in E^+$ or $e\in E^-$.
The objective is to find a partition of the nodes such that the number of \emph{disagreements}---i.e., negative edges within the same cluster and positive edges between different clusters---is minimized.
In other words, the cost of the clustering $\mathcal{C}$ is as follows:
\[\text{obj}(\mathcal{C}):=\sum_{uv \in E^+} x_{uv} + \sum_{uv \in E^-} (1 - x_{uv}),\]
where \( x_{uv} = 0 \) indicates that there exists $C\in\mathcal{C}$ such that $u,v\in C$, and \( x_{uv} = 1 \) otherwise.

CC has a standard LP relaxation leveraging the viewpoint on $x$ as a discrete metric between partitions.
Since the $x$ above satisfies the triangle inequality, we can relax the range of $x$ from $\{0,1\}$ to $[0,1]$, resulting in the following LP:
\begin{align}
\text{minimize}\quad&\sum_{uv \in E^+} x_{uv} + \sum_{uv \in E^-} (1 - x_{uv}) \label{eq:LP_CC}\tag{CC-LP}\\
\text{subject to}\quad &x_{uv} + x_{vw} \geq x_{wu}, \\
    &x_{uv} \in [0, 1].
\end{align}
The integrality gap of~\ref{eq:LP_CC} on a complete graph is known to be $2$~\cite{CHARIKAR2005360}, which indicates that the standard LP-based algorithm cannot obtain a better approximation factor below $2$.

\subsection{Pseudometric-weighted Correlation Clustering}

The \emph{weighted Correlation Clustering}  problem is a generalization of the classical CC problem in which each edge is associated with a nonnegative violation cost. Specifically, for each edge $uv$, a weight $w_{uv} \geq 0$ is provided, and violating the edge’s label (either `$+$’ or `$-$’) incurs a penalty of $w_{uv}$. This differs from the standard setting, where all violations incur a uniform cost of $1$. The weighted variant allows us to encode edge-wise reliability: when $w_{uv}$ is large, it is more reasonable to follow the given label between $u$ and $v$.

However, assuming the \emph{Unique Games Conjecture}, no $O(1)$-approximation algorithm exists for the general weighted case~\cite{DEMAINE2006172}. An exception occurs when the weight function satisfies the triangle inequality, i.e., the weights form a \emph{pseudometric}. In this \emph{pseudometric-weighted} setting, a constant-factor approximation is known~\cite{DBLP:conf/soda/CharikarG24}. Following the analysis of Charikar and Gao with $L = 2$ yields an approximation factor of
$
\overline{B}_{\text{HR}} + \frac{1}{\frac{1}{3}} \leq \frac{4}{\frac{1}{3}} + 2(L - 1) + \frac{1}{\frac{1}{3}} = 17,
$
since the second type of charge occurs at most $L - 1 = 1$ time in the charging scheme.
The following is a natural LP relaxation of the weighted CC problem, extending~\eqref{eq:LP_CC}:

\begin{align}
    \text{minimize}\quad &\sum_{uv \in E^+} w_{uv} \cdot x_{uv} + \sum_{uv \in E^-} w_{uv} \cdot (1 - x_{uv}) \label{eq:LP_wCC} \tag{wCC-LP} \\
    \text{subject to}\quad &x_{uv} + x_{vw} \geq x_{wu}, \\
    &x_{uv} \in [0, 1].
\end{align}

Here, the variable $x$ can be viewed as defining a pseudometric over the vertex set, representing the distance between clusters.
Since CC on bipartite graphs has an integrality gap of 3~\cite{chawla2015near}, and is a special case of pseudometric-weighted CC, the LP relaxation~\eqref{eq:LP_wCC} for pseudometric-weighted CC also has an integrality gap of at least 3.

\subsection{Chromatic Correlation Clustering Problem}

The \emph{Chromatic Correlation Clustering}  problem is a variant of the classical CC problem in which each cluster is additionally assigned a color~\cite{bonchi2012chromatic}. The input includes a set of $L$ possible colors, as well as a special color $\gamma$ that denotes that two vertices should not be placed in the same cluster—analogous to a negative (`$-$') edge in the classical CC setting. When $L = 1$ (i.e., a single cluster color), CCC reduces to the standard CC problem.

The following is a linear programming (LP) relaxation of the CCC problem~\cite{bonchi2012chromatic}:
\begin{align}
    \text{minimize}\quad& \sum_{\phi(uv)\neq \gamma} x_{uv}^{\phi(uv)} + \sum_{\phi(uv) = \gamma} \sum_{c \in L} (1 - x_{uv}^c) \label{eq:LP_CCC} \tag{CCC-LP} \\
    \text{subject to}\quad 
    &x_{uv}^c \geq x_u^c,\, x_v^c, \label{cond:CCC_evdom} \\
    &x_{uv}^c + x_{vw}^c \geq x_{wu}^c, \label{cond:CCC_metric} \\
    &\sum_{c \in L} x_u^c = |L| - 1, \label{cond:CCC_chroma} \\
    &x_u^c,\, x_{uv}^c \in [0,1]. \label{cond:CCC_bd}
\end{align}

Here, the variables \( x_u^c \) and \( x_{uv}^c \) are soft assignments:

\begin{itemize}
    \item \( 1-x_u^c \in [0,1] \) represents the fractional assignment of vertex \( u \) to a cluster of color \( c \).
    \item \(1- x_{uv}^c \in [0,1] \) indicates the fractional agreement between vertices \( u \) and \( v \) under color \( c \).
\end{itemize}

These variables measure the likelihood of vertices or edges being assigned to a color, with \( \{1 - x_u^c\}_{c \in L} \) forming a probability distribution over the colors assigned to vertex \( u \), subject to constraints~\eqref{cond:CCC_chroma} and~\eqref{cond:CCC_bd}.

There is also a geometric interpretation of these variables. Consider $L$ discrete pseudometric spaces $(V_c, d_c)$ where $V_c = \{u_c : u \in V\}$, and vertex $u$ is connected to $u_c$ with a link of length $x_u^c$. Then, $x_{uv}^c$ represents the \emph{bottleneck distance} between $u$ and $v$, conditioned on traversing the auxiliary connections $u \rightarrow u_c$ and $v \rightarrow v_c$. This view generalizes the classical CC setting, where the cluster-wise discrete metric can be regarded as a special case of bottleneck distances.

Since CCC generalizes the standard CC problem, the integrality gap of~\eqref{eq:LP_CCC} is at least as large as that of~\eqref{eq:LP_CC}, which is 2.

\section{Approximation Algorithm}\label{sec:algo}
Building on the LP formulations introduced in the previous sections, we now present approximation algorithms for both pseudometric-weighted CC and CCC settings.

%The algorithm for the pseudometric-weighted CC follows the regular pivot-based algorithm $\textsc{LP-Pivot}(G=(V,E=\binom{V}{2}=E^+\uplus E^-\uplus \emptyset,\,x)$~\ref{alg:pivot} with rounding functions specified in equation~\ref{eq:round_wCC}.

\begin{algorithm}
    \caption{\textsc{LP-Pivot}}
    \label{alg:pivot}
    \begin{algorithmic}
    \State \textbf{Input:} Graph $G=(V,\,E=E^+\uplus E^-\uplus E^\circ)$, LP solution $\{x_{uv}\}_{uv\in E}$.

    \State \textbf{Output:} Clustering $\mathcal{C}$ of $V$.
    \State
    
    \State Pick a pivot $v\in V$ uniformly at random.
    \State Set $C=\{v\}$.
    \For{$u\in V\backslash\{v\}$,}
        \State Set $p_{uv}$ as following:
        \[p_{uv}=\begin{cases}
            f^+(x_{uv}),&uv\in E^+;\\
            f^-(x_{uv}),&uv\in E^-;\\
            f^\circ(x_{uv}),&uv\in E^\circ.
        \end{cases}\]
        \State Update $C\leftarrow C\cup\{u\}$ with probability $1-p_{uv}$.

    \EndFor
    \State
    \Return $\mathcal{C}=\{C\}\cup\textsc{LP-Pivot}(G|_{V\backslash C},x|_{V\backslash C})$.
    \end{algorithmic}
\end{algorithm}

%\vspace{-3ex}
The algorithm for the pseudometric-weighted CC problem extends the classical LP-based pivoting method.
It uses the \textsc{LP-Pivot} algorithm, given in Algorithm~\ref{alg:pivot}, which recursively selects a pivot vertex and forms clusters via randomized rounding based on LP distances, along with the specific rounding functions~\eqref{eq:round_wCC}.
% The input graph $G=(V,E=E^+\uplus E^-\uplus E^\circ)$ is partitioned by edge type, and the rounding probabilities—defined in equation~\ref{eq:round_wCC}—are adapted accordingly.
For each edge $uv$ incident to the pivot $v$, the probability of $u$ being merged with $v$ into a same cluster is determined by not only on the LP solution $x_{uv}$, but also on the edge’s type as specified in the input.
% The type of the edge affects by choosing what rounding function to adapt on the LP solution.
% However, in the case of pseudometric-weighted CC, all types of edges use the same rounding function (equation~\ref{eq:round_wCC}), so the edge type has no effect on the rounding.
Rounding functions translate the LP solution ${x_{uv}}$ to a rounding probability to guarantee a constant-factor approximation by controlling the expected clustering cost.

The \textsc{LP-CCC} algorithm for CCC, given in Algorithm~\ref{alg:CCC}, first partitions the vertices according to their LP-derived color distributions.
It then applies the \textsc{LP-Pivot} algorithm, with edge types partitioned by color.
CCC-specific rounding functions (equations~\ref{eq:round_CCC_pm} and \ref{eq:round_CCC_0}) are used to account for the color structure, and the final clustering is obtained by combining the color class clusters.

\begin{algorithm}
    \caption{\textsc{LP-CCC}}
    \label{alg:CCC}
    \begin{algorithmic}
    \State \textbf{Input:} Graph $G=(V,\,E)$, color function $\phi:E\rightarrow L\cup\{\gamma\}$, LP solution $\{x_u^c\}_{u\in V,\,c\in L}$ and $\{x_{uv}^c\}_{uv\in E,\,c\in L}$.

    \State \textbf{Output:} Clustering $\mathcal{C}$ of $V$, Coloring function $\Phi:\mathcal{C}\rightarrow L$.

    \State 
    \State Initialize $\mathcal{C}=\emptyset,\,S_c=\emptyset$ for all $c\in L$.
    \For{$u\in V$}
        \If{$\exists c\in L$ s.t.\ $x_u^c<\frac{1}{2}$,}
            \State Update $S_c\leftarrow S_c\cup\{u\}$.
        \Else
            \State Update $\mathcal{C}\leftarrow \mathcal{C} \cup\{\{u\}\}$.% This indicate a singleton cluster.
            \State Assign $\Phi(\{u\})$ as an arbitrary color.
        \EndIf
    \EndFor
    \For{$c\in L$}
        \State $G_c=(S_c,\,E_c = E_c^+\uplus E_c^- \uplus E_c^\circ )$, where $E_c= E\cap\binom{S_c}{2}$,
        \State and $E^+\uplus E^-\uplus E^\circ$ is defined as a partition by color $c,\,\gamma,L\backslash\{c\}$ respectively.
        \State Set $\mathcal{C}_c=\textsc{LP-Pivot}(G_c,x^c|_{E_c})$.
        \State Update $\mathcal{C}\leftarrow \mathcal{C}\cup \mathcal{C}_c$.
        \State Assign $\Phi(C)=c$ for all $C\in\mathcal{C}_c$.
    \EndFor
    \State \Return $\mathcal{C},\,\Phi$.
    \end{algorithmic}
\end{algorithm}

\section{Rounding Functions}
\label{sec:round}

The effectiveness of the \textsc{LP-Pivot} and \textsc{LP-CCC} algorithms critically depends on the choice of rounding functions used in the clustering process.
Rounding functions $f^+,\,f^-,\,f^\circ:[0,1]\rightarrow [0,1]$ convert the LP value $x_{uv}$ to the non-selection probability $p_{uv}$~\cite{chawla2015near}. The sign of the edge $uv$—either `$+$', `$-$', or `$\circ$'—determines which rounding function is applied. The sign `$\circ$' indicates that the edge does not belong to $E$.
The following natural conditions are imposed on any rounding function $f$:
\begin{align}
    &f(0)=0,\,f(1)=1;\label{cond:round_bd}\\
    &x<y\Rightarrow f(x)\leq f(y).\label{cond:round_nondec}
\end{align}
Condition~\ref{cond:round_bd} is not only intuitive but also necessary in certain cases, such as ensuring $f^+(0)=0$ and $f^-(1)=1$. Other constraints are not required in the proofs of Theorems~\ref{thm:wCC_lb} and \ref{thm:CCC_lb}, which thus provide lower bounds on the approximation factors for \emph{`general'} rounding functions. 

\begin{lemma}\label{lem:lp-pivot}
The \textsc{LP-Pivot} algorithm achieves a constant-factor approximation in expectation only if $f^+(0)=0$ and $f^-(1)=1$. 
\end{lemma}

\begin{proof}
    We prove it by contradiction on some graph instances.
    \\
    \\
    \textbf{Case 1} $f^+(0)=0.$ Consider $G=(V,E=\binom{V}{2}=E\uplus \emptyset \uplus \emptyset)$.
    The optimal clustering is $\mathcal{C}^*=\{V\}$, satisfying $\text{obj}(\mathcal{C}^*)=0$, and the optimal LP solution is $x^*\equiv 0$.
    
    Suppose $f^+(0)>0$. Then $\Pr[\textsc{LP-Pivot}(G,0)\neq \mathcal{C}^*]>0$.
    Since $\text{obj}(\mathcal{C})>0$ if and only if $\mathcal{C}\neq\mathcal{C}^*$, this leads to a contradiction with the assumption of the expected constant factor approximation.
\\
\\
    \textbf{Case 2} $f^-(1)=1.$
    Consider $G=(V,E=\binom{V}{2}=\emptyset\uplus E \uplus \emptyset)$.
    The optimal clustering is $\mathcal{C}^*=\{\{v\}:v\in V\}$.
    The following arguments are similar to Case 1.
\end{proof}

Different variants of the CC problem may use different rounding functions. In this paper, we provide rounding functions for both the pseudometric-weighted CC and CCC problems.

\subsection{Pseudometric-weighted Correlation Clustering}

We propose the following  rounding functions that yield a tight approximation factor:
\begin{equation}
\label{eq:round_wCC}
    f^+(x)=f^-(x)=\begin{cases}
        0,&x<0.4;\\
        \frac{5}{3}x,&0.4\leq x<0.6;\\
        1,&x\geq 0.6.
    \end{cases}
\end{equation}

% \begin{figure}
%     \centering
%     \includegraphics[width=0.4\linewidth]{fig/pwCC_new.png}
%     \caption{Rounding function for the pseudometric-weighted CC problem.}
%     \label{fig:round_wCC}
% \end{figure}

With these functions, the algorithm achieves an expected approximation factor of $10/3$. Moreover, no other rounding function can improve this factor, as shown in Section~\ref{subsec:analysis_wCC}.

\subsection{Chromatic Correlation Clustering}

We further consider the rounding functions $f^+,\,f^-$ from Chawla et al.~\cite{chawla2015near}, which yield a $2.06$-approximation for classical CC, and introduce a new function $f^\circ$ to handle $\circ$-edges for CCC:
\begin{equation}
\label{eq:round_CCC_pm}
    f^+(x)=\begin{cases}
        0,&x<0.19;\\
        \left(\frac{x-0.19}{0.5095-0.19}\right)^2,&0.19\leq x< 0.5095;\\
        1,&x\geq 0.5095,
    \end{cases}
    \quad
    f^-(x)=x,
\end{equation}
\begin{equation}
\label{eq:round_CCC_0}
    % f^\circ(x)=\begin{cases}
    %     0,&x<0.4;\\
    %     0.7902,&0.4\leq x<0.49446;\\
    %     0.3(x-0.5)+0.85,&x\geq 0.49446.
    % \end{cases}
    f^\circ(x)=\begin{cases}
        1.7x,&x<0.5;\\
        0.3x+0.7,&x\geq 0.5.
    \end{cases}
\end{equation}

This rather intricate function was computed to not intersect with analytic bounds that violate an approximation factor of $\alpha=2.15$, as illustrated in Figure~\ref{fig:round_CCC_region}.

\begin{figure}
    \centering
    \includegraphics[height=0.40\linewidth]{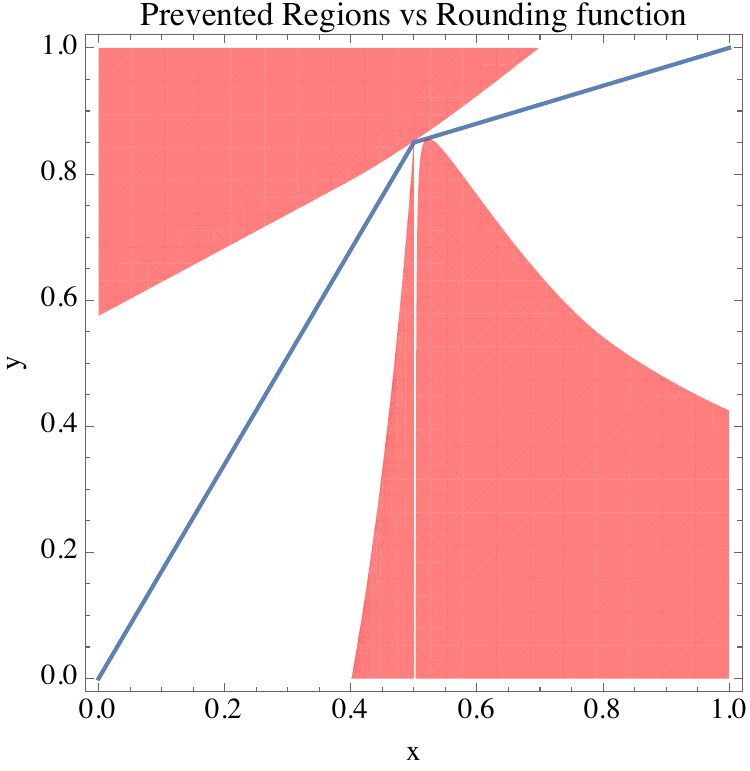}
   \caption{The region where $f^\circ$ violates the $\alpha = 2.15$-approximation for CCC using the proposed $f^\circ$ defined in~\eqref{eq:round_CCC_0}.}

    \label{fig:round_CCC_region}
\end{figure}

\subsection{Comparison with Prior Work}
\textbf{Pseudometric-weighted CC:} The \textsc{LP-UMVD-Pivot} algorithm, recently proposed by Charikar and Gao~\cite{DBLP:conf/soda/CharikarG24} for the Ultrametric Violation Distance (UMVD) problem, follows a pivoting-based rounding strategy applied to an LP relaxation. When the number of distinct pairwise distances between elements, denoted by \( L \), is equal to 2, their algorithm can be viewed as a special case of our \textsc{LP-Pivot} framework. In this setting, the two distances \( d_1 \) and \( d_2 \) correspond to the `$-$' and `$+$' labels, respectively, used in our rounding procedure.
 The rounding functions used are:
\[
f^+(x)=f^-(x)=\begin{cases}
    0,&x<\alpha;\\
    \frac{\max\{x-\alpha\beta,\,0\}}{1-\alpha\beta},&\alpha\leq x\leq 1-\alpha;\\
    1,&x>1-\alpha,
\end{cases}
\qquad
f^\circ(x)=\begin{cases}
    0,&x<\alpha\beta;\\
    x,&\alpha\beta \leq x\leq 1-\alpha\beta;\\
    1,&x>1-\alpha\beta.
\end{cases}
\]

Here, $\alpha$ and $\beta$ are fixed algorithmic parameters. For the pseudometric-weighted CC problem, they are set to $\alpha=\frac{1}{3}$ and $\beta=0$, and $f^\circ$ is unused.

With triple-based analysis, this choice yields an approximation factor of $6$, as shown in the subsection~\ref{sec:wCC_old}.

\textbf{Chromatic CC:} The LP-based pivoting algorithm by Xiu et al.~\cite{XiuHTCH22} uses LP values directly as probabilities, corresponding to the following rounding functions:
\[
f^+(x)=f^-(x)=f^\circ(x)=x.
\]
This setting is known to achieve an approximation factor of $2.5$.

\subsection{Proof for $6$-approximation on Pseudometric-weighted CC with \textsc{LP-UMVD-Pivot}}
\label{sec:wCC_old}
With the parameters $\alpha=\frac{1}{3}$ and $\beta=0$, the corresponding rounding functions are as follows:
\begin{equation}
\label{eq:round_wCC_old}
    f^+(x)=f^-(x)=\begin{cases}
    0,&x<\frac{1}{3};\\
    x,&\frac{1}{3} \leq x\leq \frac{2}{3};\\
    1,&x>\frac{2}{3}.
\end{cases}
\end{equation}

Since the formulas for $e.cost$ and $e.lp$ in the pseudometric-weighted CC are compatible with those in the classical CC, and weights $w$ affect $\mathcal{C}$ linearly, the following lemma presented by Chawla et al.~\cite{chawla2015near} remains applicable.
\begin{lemma}[Lemma 5 of \cite{chawla2015near}]\label{lem:vd_bd}
    Suppose $f^+$ is a monotonically non-decreasing piecewise convex function; $f^-$ is a monotonically non-decreasing piecewise concave function. Then $\mathcal{C}\geq 0$ for all possible configurations if it holds whenever:
    \begin{enumerate}
        \item the triangle inequality is tight for $(x_{uv},x_{vw},x_{wu})$; or
        \item each value of $x_e\:(e\in\{uv,vw,wu\})$ belongs to the endpoint of the domain for some piece of the rounding function $f^s$, where $s$ is the sign of $e$.
    \end{enumerate}
\end{lemma}
Therefore, it suffices to show $6\cdot LP-ALG \geq 0$ in either of the two cases outlined above.

We introduce another useful lemma:
\begin{lemma}
\label{lem:xyz}
    If $x_{uv},x_{vw},x_{wu}\in \left[\frac{1}{\alpha},1-\frac{1}{\alpha}\right]$, then $\alpha\cdot LP-ALG\geq 0$.
\end{lemma}
\begin{proof}
    \begin{align*}
        \alpha\cdot e.lp_w-e.cost_w&\geq \alpha(1-p_{vw}p_{wu})\min\{x,1-x\}\\
        &\quad -\max\{p_{vw}(1-p_{wu})+(1-p_{vw})p_{wu},(1-p_{vw})(1-p_{wu})\}\\
        &\geq \alpha(1-p_{vw}p_{wu})\cdot\frac{1}{\alpha}-(1-p_{vw}p_{wu})= 0.
    \end{align*}
    Same inequality holds for other vertices, hence $\alpha\cdot LP-ALG\geq 0$.
\end{proof}

For notational simplicity, define $(x,y,z):=(x_{uv},x_{vw},x_{wu})$ and \emph{w.l.o.g.}\ $x\leq y\leq z$.

\subsubsection{Triangle inequality is tight.}
In this case, $z=x+y$.
As the rounding functions are defined piecewise, we can partition the entire space of possible configuration of $(x,y)$ into a finite number of regions (labeled I through VI in Figure~\ref{fig:wCC_old_region}) where the behaviors of $p_{uv},\,p_{vw},\,p_{wu}$ are well identified in each.
\begin{figure}
    \centering
    \includegraphics[width=0.4\linewidth]{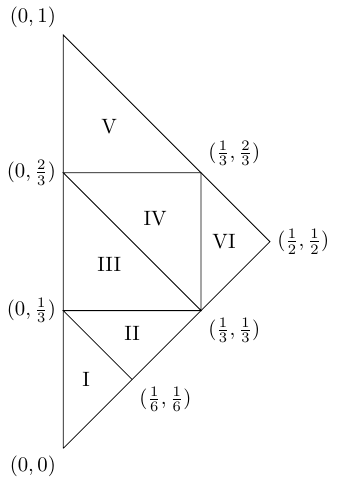}
    \caption{Configuration of regions induced from~\eqref{eq:round_wCC_old}.}
    \label{fig:wCC_old_region}
\end{figure}

Since $f^+=f^-$, the value of $p_e$ is independent of the sign of $e$, therefore the formula for both $e.cost_w$ and $e.lp_w$ are independent of the sign of the edges $vw$ and $wu$.
Therefore, once the region and the sign of the edge $uv$ are fixed, the formula for $e.cost_w$ and $e.lp_w$ are fully resolved as a function of $x,y,z$, as shown in Table~\ref{tab:wCC_old}.
\begin{table}
\caption{Formula for $e.cost$ and $e.lp$ for each configuration.}
\label{tab:wCC_old}

\begin{tabular}{llrrrrrr}
\toprule
 &  & \multicolumn{2}{c}{$x_{uv}=x$} & \multicolumn{2}{c}{$x_{vw}=y$} & \multicolumn{2}{c}{$x_{wu}=z=x+y$} \\
\cmidrule(r){3-4}\cmidrule(r){5-6}\cmidrule(r){7-8}
Region       & Sign     & $e.cost_w$        & $e.lp_w$       & $e.cost_u$        & $e.lp_u$       & $e.cost_v$          & $e.lp_v$         \\
       \midrule
I      & `$+$'  & $0$ &  $x$   & $0$ &  $y$  &  $0$  & $z$   \\
       & `$-$'  & $1$ & $1-x$ &   $1$  &  $1-y$  & $1$ & $1-z$  \\
II     & `$+$'  & $z$ & $x$   &  $z$   &$y$&$0$&$z$\\
       & `$-$'  &$1-z$  &$1-x$&   $1-z$ &$1-y$&$1$&$1-z$\\
III    & `$+$'  &$y+z-2yz$ &  $(1-yz)x$  &  $z$ &  $y$&$y$&  $z$\\
       & `$-$'  &$(1-y)(1-z)$  &  $(1-yz)(1-x)$ & $1-z$&$1-y$&      $1-y$&  $1-z$\\
IV      & `$+$'  &     $1-y$&  $(1-y)x$&  $1$&   $y$&  $y$&  $z$\\
       & `$-$'  &$0$&  $(1-y)(1-x)$&    $0$&  $1-y$&  $1-y$&$1-z$\\
V     & `$+$'  &$0$& $0$&$1$&$y$&$1$& $z$\\
       & `$-$'  &$0$& $0$&$0$&$1-y$&   $0$&  $1-z$\\ 
VI     & `$+$'  &$1-y$&  $(1-y)x$& $1-x$&$(1-x)y$& $x+y-2xy$&    $(1-xy)z$\\
       & `$-$'  &$0$&   $(1-y)(1-x)$&   $0$&  $(1-x)(1-y)$&    $(1-x)(1-y)$&    $(1-xy)(1-z)$\\
       \bottomrule
\end{tabular}
\end{table}
By the implication of Lemma~\ref{lem:pseudo}, it suffices to verify all possible configurations of the following procedure:
\begin{enumerate}
    \item \textbf{Select any region}: This results in $12$ corresponding cells, with $6$ possible cases.
    \item \textbf{Select any $2$ columns}: This results in $8$ corresponding cells, with $3$ possible cases.
    \item \textbf{For each column, select the sign}: This results in $4$ corresponding cells, with $4$ possible cases.
    \item \textbf{The following must hold}: $\alpha\cdot(e.lp_a+e.lp_b)-(e.cost_a+e.cost_b)\geq 0$, where $a,b$ are the vertices corresponding to the selected columns.
\end{enumerate}

\subsubsection*{Region I, II, V}
All formulas are affine; therefore, checking $\alpha\cdot(e.lp_a+e.lp_b)-(e.cost_a+e.cost_b)\geq 0$ at the extremal points is sufficient. Table~\ref{tab:wCC_old_extremal} shows extremal points of the regions with the value of $6\cdot e.cost-e.lp$ for each configuration. In each region with the point, adding any two values from different columns results in a nonnegative value, validating the nonnegativity throughout the region with sign configuration.

\begin{table}[]
\caption{Value of $6\cdot e.lp-e.cost$ for each extremal point of the region.}
    \label{tab:wCC_old_extremal}
    \centering
\begin{tabular}{lllrrr}
\toprule
Region&$(x,y)$&Sign& $6\cdot e.lp_w-e.cost_w$& $6\cdot e.lp_u-e.cost_u$&$6\cdot e.lp_v-e.cost_v$\\
\midrule
I&$(0,0)$&`$+$'&$0$&$0$&$0$\\
&&`$-$'&$5$&$5$&$5$\\
&$(0,\frac{1}{3})$&`$+$'&$0$&$2$&$2$\\
&&`$-$'&$5$&$3$&$3$\\
&$(\frac{1}{6},\frac{1}{6})$&`$+$'&$1$&$1$&$2$\\
&&`$-$'&$4$&$4$&$3$\\
II&$(0,\frac{1}{3})$&`$+$'&$-1/3$&$5/3$&$2$\\
&&`$-$'&$16/3$&$10/3$&$3$\\
&$(\frac{1}{6},\frac{1}{6})$&`$+$'&$2/3$&$2/3$&$2$\\
&&`$-$'&$13/3$&$1/3$&$3$\\
&$(\frac{1}{3},\frac{1}{3})$&`$+$'&$4/3$&$4/3$&$4$\\
&&`$-$'&$11/3$&$11/3$&$1$\\
V&$(0,\frac{2}{3})$&`$+$'&$0$&$3$&$3$\\
&&`$-$'&$0$&$2$&$2$\\
&$(0,1)$&`$+$'&$0$&$5$&$5$\\
&&`$-$'&$0$&$0$&$0$\\
&$(\frac{1}{3},\frac{2}{3})$&`$+$'&$0$&$3$&$5$\\
&&`$-$'&$0$&$2$&$0$\\
\bottomrule
\end{tabular}
\end{table}
\subsubsection*{Region III}
\label{subsubsec:old_III}
We can also further optimize procedures 3 and 4: Compute $\alpha\cdot e.lp_a-e.cost_a$ for each region, sign, and vertex configuration.
For a given region and vertex, if one sign configuration dominates another, choose the smaller one.

For example, consider $6\cdot e.lp_w-e.cost_w$ in Table~\ref{tab:wCC_old}.
Since $x\leq \frac{1}{2}$, $(1-yz)(1-x)\geq (1-yz)x$.
For $e.cost_w$, $y+z-2yz\geq (1-y)(1-z)$ since
\[
    y+z-2yz-(1-y)(1-z)=\frac{1-(2-3y)(2-3z)}{3}\geq 0
\]
since $y,z\in [1/3,2/3]$ within the region.
Therefore, $6\cdot e.lp_w-e.cost_w$ is smaller with the sign `$+$'.

Now, consider the minimum value possible for the formula above within the region.
Replacing $x=z-y$,
\[6\cdot e.lp_w-e.cost_w=-6yz^2+6y^2z+2yz-7y+5z.\]
Fixing $y$, the formula is concave of $z$; therefore, it is minimized in $z=y$ or $z=2/3$, resulting in $2y^2-2y$ and $4y^2-\frac{25}{3}y+\frac{10}{3}$ each.
Within the range $y\in [1/3,2/3]$, the minimum value of each is $-1/2$ in $y=1/2$ and $-4/9$ in $y=2/3$.
Thus, $6\cdot e.lp_w-e.cost_w\geq -1/2$.

Finally, consider the minimum value possible for $6\cdot e.lp_u-e.cost_u$ and $6\cdot e.lp_v-e.cost_v$. Again, these are affine; hence checking on extremal points is sufficient.
It turns out that $6\cdot e.lp_u-e.cost_u\geq 4/3$, minimized at $(y,z)=(1/3,2/3)$ with `$+$' sign; $6\cdot e.lp_v-e.cost_v\geq 4/3$, minimized at $(y,z)=(1/3,2/3)$ with `$-$' sign.

Therefore, any objective value is greater than or equal to $4/3-1/2=5/6\geq0$ in the region.

\subsubsection*{Region IV}
\label{subsubsec:old_IV}
We can utilize the optimization method explained in \ref{subsubsec:old_III} again: Since $(1-y)x\leq (1-y)(1-x)$ and $1-y\geq 0$, the formula with the sign `$+$' is smaller for $6\cdot e.lp_w-e.cost_w$, which is $(6x-1)(1-y)$. This has a minimum value of $-1/3$ in $(x,y)=(0,2/3)$.

Also, since $6z-y=6x+5y\geq 10/3> 5/2$ in this region, the formula with the sign `$-$' is smaller for $6\cdot e.lp_v-e.cost_v$, which is $5-6x-5y$. This has a minimum value of $-1/3$ in $(x,y)=(1/3,2/3)$.

The formula for $6\cdot e.lp_u-e.cost_u$ is $6y-1$ or $6(1-y)$. Each of them has a minimum value $1$ in $(x,y)=(1/3,1/3)$, $2$ in $(x,y)=(0,2/3)$, respectively.
Thus, $(6\cdot e.lp_w-e.cost_w)+(6\cdot e.lp_u-e.cost_u)$ and $(6\cdot e.lp_u-e.cost_u)+(6\cdot e.lp_v-e.cost_v)$ are at least $-1/3+1=2/3\geq 0$.

Finally,
\begin{align*}(6\cdot e.lp_w-e.cost_w)+(6\cdot e.lp_v-e.cost_v)&\geq (6x-1)(1-y)+(5-6x-5y)\\
&=-6xy-4y+4\geq -6y+4\geq 0,
\end{align*}
Making the approximation factor $\alpha=6$ tight at $(x,y)=(1/3,2/3)$ as well as the objective value nonnegative.

\subsubsection*{Region VI}
\label{subsubsec:old_VI}
As in IV, the formula with sign `$+$' is smaller for $6\cdot e.lp_w-e.cost_w$, which is $(6x-1)(1-y)$. This has a minimum value of $1/3$ in $(x,y)=(1/3,2/3)$.

The formula for $6\cdot e.lp_u-e.cost_u$ is $(1-x)(6y-1)$ or $6(1-x)(1-y)$.
Each of them has a minimum value $2/3$ in $(x,y)=(1/3,1/3)$, $4/3$ in $(x,y)=(1/3,2/3)$, respectively.

For $6\cdot e.lp_v-e.cost_v$, consider the difference between those of `$+$' sign and `$-$' sign.
\begin{align*}
    6(1-xy)(z-(1-z))-(x+y-2xy-(1-x)(1-y))
    &\geq 2(1-xy)-(-3xy+2x+2y-1)\\
    &=xy-2x-2y+3\\
    &=(2-x)(2-y)-1\geq \frac{11}{9}>0,
\end{align*}
since $z\geq 2/3$.
Therefore, the formula with sign `$-$' is smaller, which is $6(1-xy)(1-z)-(1-x)(1-y)$.

Define $p=xy$, then
\[6(1-xy)(1-z)-(1-x)(1-y)=6pz-5z-7p+5,\]
with region $z\in [2/3,1],\,p\in\left[\frac{z}{3}-\frac{1}{9},\left(\frac{z}{2}\right)^2\right]$.
Fixing the value of $z$, it is a decreasing function with respect to $p$; fixing the value of $p$, it is a decreasing function with respect to $z$. Therefore, the function is minimized at $z=1,\,p=\frac{1}{4}$ with the value $-\frac{1}{4}$.

Therefore, any objective value is greater than or equal to $1/3-1/4=1/12\geq 0$ in the region.

Overall, $6\cdot LP - ALG\geq 0$ when the triangle inequality is tight.

\subsubsection{All coordinates belong to endpoint of segments.}
There are total $4$ possible LP values: $0,\,1/3,\,2/3,\,1$.
Among all the possible $(x,y,z)$ tuples, there are only $7$ tuples whose triangle inequality is not tight.
Moreover, $3$ tuples $(x,y,z)=(\frac{1}{3},\frac{1}{3},\frac{1}{3}),\,(\frac{1}{3},\frac{2}{3},\frac{2}{3}),\,(\frac{2}{3},\frac{2}{3},\frac{2}{3})$ among them only consist of value $1/3$ or $2/3$, thus $3\cdot LP-ALG\geq0$ by Lemma~\ref{lem:xyz}.
Therefore, we only need to verify the remaining $3$ tuples.

Since left-side and right-side limits of the value of rounding function at such point might differ, we have to consider all possible limits of $6\cdot LP-ALG$ up to $2^3=8$ directions.
Table~\ref{tab:wCC_bd_old} shows remaining $3$ tuples and the directional limit of the value of $6\cdot e.cost - e.lp$ for each configuration.
Note that since $x\leq y\leq z$, only up to $4$ out of $8$ directions are required to verification.
\begin{table}[]
\caption{Value of $6\cdot e.lp-e.cost$ for each points in consideration of Case 2.}
    \label{tab:wCC_bd_old}
    \centering
\begin{tabular}{llrrr}
\toprule
$(x,y,z)$&Sign& $6\cdot e.lp_w-e.cost_w$& $6\cdot e.lp_u-e.cost_u$&$6\cdot e.lp_v-e.cost_v$\\
\midrule
$(\frac{2}{3}-\delta,\frac{2}{3}-\delta,1)$&`$+$'&$1$&$1$&$26/9$\\
&`$-$'&$2/3$&$2/3$&$-1/9$\\
$(\frac{2}{3}-\delta,\frac{2}{3}+\delta,1)$&`$+$'&$0$&$1$&$5/3$\\
&`$-$'&$0$&$2/3$&$0$\\
$(\frac{2}{3}+\delta,\frac{2}{3}+\delta,1)$&`$+$'&$0$&$0$&$0$\\
&`$-$'&$0$&$0$&$0$\\
$(\frac{1}{3}-\delta,1,1)$&`$+$'&$0$&$5$&$5$\\
&`$-$'&$0$&$0$&$0$\\
$(\frac{1}{3}+\delta,1,1)$&`$+$'&$0$&$10/3$&$10/3$\\
&`$-$'&$0$&$0$&$0$\\
$(\frac{2}{3}-\delta,1,1)$&`$+$'&$0$&$5/3$&$5/3$\\
&`$-$'&$0$&$0$&$0$\\
$(\frac{2}{3}+\delta,1,1)$&`$+$'&$0$&$0$&$0$\\
&`$-$'&$0$&$0$&$0$\\
$(1,1,1)$&`$+$'&$0$&$0$&$0$\\
&`$-$'&$0$&$0$&$0$\\
\bottomrule
\end{tabular}
\end{table}
As in argument on region I, II, V of Case 1, in each point with direction specified, adding any two values with different column results in a nonnegative value, validating the nonnegativity.
\\
\\
This completes our case analysis.

\section{Triple-based Analysis}
\label{sec:analysis}

To complete the analysis of the algorithm, it suffices to show that for every triple of vertices \( u, v, w \in V \), the expected cost incurred by the algorithm, denoted \( ALG(uvw) \), is at most a factor \( \alpha \) times the corresponding LP cost \( LP(uvw) \). That is,
\[
ALG(uvw) \leq \alpha \cdot LP(uvw).
\]
%This key inequality enables us to upper bound the expected cost of the algorithm at each step \( t \), denoted \( \mathbb{E}[\text{ALG}_t] \), in terms of the expected LP cost \( \mathbb{E}[\text{LP}_t] \). By linearity of expectation and summing over all steps, this yields an overall approximation factor of \( \alpha \).

If the inequality holds for every triple, then the total expected cost of the algorithm is at most \( \alpha \cdot LP \). To show this, the analysis expresses the expected algorithmic cost and LP cost as averages over all possible pivot choices and vertex triples. Specifically, it defines:
\begin{itemize}
    \item \( e.cost_w(u, v) \): the expected cost of violating constraint \( (u, v) \), conditioned on pivot \( w \),
    \item \( e.lp_w(u, v) \): the expected LP \emph{charge} of edge \( (u, v) \), conditioned on pivot \( w \).
\end{itemize}

\subsection{Pseudometric-weighted Correlation Clustering}
\label{subsec:analysis_wCC}

In the CC setting, we use the function $\mathcal{C}$, as defined in~\cite{chawla2015near}, to measure the gap:
\[
\mathcal{C}(x_{uv},x_{vw},x_{wu},p_{uv},p_{vw},p_{wu}) = \alpha \cdot LP(uvw) - ALG(uvw),
\]
where
\begin{align*}
ALG(uvw)&=e.cost_w(uv)+e.cost_u(vw)+e.cost_v(wu),\\
LP(uvw)&= e.lp_w(uv)+ e.lp_u(vw)+ e.lp_v(wu),
\end{align*}
and
\begin{align*}
e.cost_w(u,v) &= \begin{cases}
    p_{uw}(1 - p_{vw}) + (1 - p_{uw})p_{vw}, & uv \in E^+; \\
    (1 - p_{uw})(1 - p_{vw}), & uv \in E^-,
\end{cases}\\
e.lp_w(u,v) &= \begin{cases}
    (1 - p_{uw}p_{vw})x_{uv}, & uv \in E^+; \\
    (1 - p_{uw}p_{vw})(1 - x_{uv}), & uv \in E^-.
\end{cases}
\end{align*}

In the weighted CC setting, edge weights further influence the value of $\mathcal{C}$:
\[
\mathcal{C}(x_{uv},x_{vw},x_{wu},p_{uv},p_{vw},p_{wu},w_{uv},w_{vw},w_{wu}) = \alpha \cdot LP(uvw) - ALG(uvw),
\]
with the definition for $e.cost$ and $e.lp$ remains the same; the classical CC corresponds to $(w_{uv},w_{vw},w_{wu}) = (1,1,1)$.

Under the pseudometric constraint on weights $w$, we can reduce the number of cases to consider in the analysis.

\begin{lemma}\label{lem:pseudo}
If $\alpha \cdot LP(uvw) - ALG(uvw) \geq 0$ holds for weight configurations $(w_{uv},w_{vw},w_{wu}) \in \{(1,1,0),\,(1,0,1),\,(0,1,1)\}$, then the inequality also holds for any configuration $(w_{uv},w_{vw},w_{wu})$ satisfying the triangle inequality.
\end{lemma}
\begin{proof}
    Let all $x_{uv},\,x_{vw},\,x_{wu},\,p_{uv},\,p_{vw},\,p_{wu}$ be fixed.
    $ALG(uvw)$ and $LP(uvw)$ can be written as
    \[ALG(uvw)=w_{uv}\cdot e.cost_w(uv)+w_{vw}\cdot e.cost_u(vw)+w_{wu}\cdot e.cost_v(wu)\]
    and
    \[LP(uvw)=w_{uv}\cdot e.lp_w(uv)+w_{vw}\cdot e.lp_u(vw)+w_{wu}\cdot e.lp_v(wu).\]

    Therefore, the function $\alpha LP(uvw) - ALG(uvw)$ is linear w.r.p.\ $(w_{uv},w_{vw},w_{wu})$.

    Since the set of $(w_{uv},w_{vw},w_{wu})$ that satisfies the triangle inequality forms a convex cone generated by $(1,1,0),\,(1,0,1),\,(0,1,1)$,
    the function value is nonnegative for all such $(w_{uv},w_{vw},w_{wu})$ if and only if the value is nonnegative for $(w_{uv},w_{vw},w_{wu})\in\{(1,1,0),\,(1,0,1),\,(0,1,1)\}$.
\end{proof}
This lemma implies that the algorithm achieves an approximation factor of $\alpha$ if all of the following inequalities are satisfied for every possible configuration on the triangle $uvw$:
\begin{align*}
    e.cost_w(uv) + e.cost_u(vw) &\leq \alpha \cdot (e.lp_w(uv) + e.lp_u(vw)), \\
    e.cost_w(uv) + e.cost_v(wu) &\leq \alpha \cdot (e.lp_w(uv) + e.lp_v(wu)), \\
    e.cost_u(vw) + e.cost_v(wu) &\leq \alpha \cdot (e.lp_u(vw) + e.lp_v(wu)).
\end{align*}

We obtain a lower bound on the approximation factor of \textsc{LP-Pivot} by verifying the feasibility of rounding functions that satisfy the above inequalities. To this end, we analyze several configurations of LP values and edge signs on triangle $uvw$.

In Theorems~\ref{thm:wCC_lb} and~\ref{thm:CCC_lb}, the notation `$(a,b,c)$ with $(s_1,s_2,s_3)$' denotes $(x_{uv},x_{vw},x_{wu}) = (a,b,c)$, where each edge sign is given by $s_1$, $s_2$, and $s_3$, respectively.

\begin{theorem}
\label{thm:wCC_lb}
The lower bound on the approximation factor of \textsc{LP-Pivot} in pseudometric-weighted correlation clustering is $10/3$.   
\end{theorem}

\begin{proof}
Suppose that \textsc{LP-Pivot} achieves an expected $\alpha$-approximation.

We begin by considering the configuration $(x,1-x,1)$ with signs $(+,+,-)$ or $(+,-,-)$ and weights $(1,0,1)$. Then:
\begin{align}
    \alpha (1 - f^+(1 - x))x &\geq (2 - f^+(x))(1 - f^+(1 - x)), \\
    \alpha (1 - f^-(1 - x))x &\geq (2 - f^+(x))(1 - f^-(1 - x)),
\end{align}
which implies:
\begin{equation}
\label{eq:weight_lb_upper}
    f^+(x) \leq 2 - \alpha x \quad \implies \quad f^+(1 - x) = f^-(1 - x) = 1.
\end{equation}

Next, consider the configuration $(0,x,x)$ with signs $(+,+,+)$ and weights $(1,1,0)$:
\begin{equation}
\label{eq:weight_lb_lower1}
    \alpha x\geq 3f^+(x) - 2f^+(x)^2,
\end{equation}
which yields the bound $f^+(x) \leq \alpha x$. Therefore, for $x \in [0, 1/\alpha]$, we have
\[
f^+(x) \leq \alpha x \leq 2 - \alpha x, \quad f^+(1 - x) = 1.
\]

Now consider the configuration $(x,x,2x)$ with signs $(+,+,+)$ and weights $(1,1,0)$:
\begin{equation}
\label{eq:weight_lb_lower2_1}
    \alpha x (1 - f^+(x) f^+(2x)) \geq f^+(x) + f^+(2x) - 2f^+(x) f^+(2x).
\end{equation}

Rewriting the inequality, we obtain:
\begin{equation}
    (1 - (2 - \alpha x) f^+(x)) f^+(2x) \leq \alpha x - f^+(x).
\end{equation}

Observe that for $x \in [0,1/\alpha)$ where $f^+(x) < 1$, the term
\[
1 - (2 - \alpha x) f^+(x) \geq 1 - 2f^+(x) + f^+(x)^2 \geq 0
\]
is strictly positive. Hence:
\begin{equation}
\label{eq:weight_lb_lower2_2}
    f^+(2x) \leq \frac{\alpha x - f^+(x)}{1 - (2 - \alpha x) f^+(x)} 
    = \alpha x - \frac{(1 - \alpha x)^2 f^+(x)}{1 - (2 - \alpha x) f^+(x)} 
    \leq \alpha x,
\end{equation}
for $x \in [0, 1/\alpha)$.

Combining the above, we conclude:
\[
f^+(x) < \frac{\alpha}{2} x \quad \text{for } x \in [0, 2/\alpha).
\]

Set \( x = \frac{4}{3\alpha} \). Then from~\eqref{eq:weight_lb_upper} and~\eqref{eq:weight_lb_lower2_2}, we have:
\[
f^+\left(\frac{4}{3\alpha}\right) \leq \frac{2}{3} = 2 - \alpha x, \quad \text{so} \quad f^+\left(1 - \frac{4}{3\alpha}\right) = 1.
\]

Since $f^+(x) < 1$ for $x \in [0, 2/\alpha)$, this implies:
\[
\frac{2}{\alpha} \leq 1 - \frac{4}{3\alpha} \quad \Longrightarrow \quad \alpha \geq \frac{10}{3}.
\]

Hence, the approximation factor must be at least $10/3$.
\end{proof}

Conversely, there exist rounding functions $f^+,\,f^- $ making the approximation factor of \textsc{LP-Pivot} by $10/3$, providing that the lower bound above is tight.
\begin{theorem}
\label{thm:wCC_alg}
The \textsc{LP-Pivot} algorithm with the rounding function defined in equation~\ref{eq:round_wCC} yields a $10/3$-approximation algorithm for pseudometric-weighted CC.     The proof is provided in the Appendix.
\end{theorem}

\subsection{Chromatic Correlation Clustering}
\label{subsec:analysis_CCC}

We analyze the performance of the \textsc{LP-CCC} algorithm. This algorithm begins by assigning each vertex to its majority color based on the LP solution, followed by a pivot-based clustering routine.

Due to the strict majority condition, any edge not included in $\biguplus E_c$ must have an LP value of at least $1/2$. Thus, the cost incurred by such edges is at most twice their LP contribution~\cite{XiuHTCH22}.

Within each color class $S_c$, corresponding to color $c$, we follow an analysis similar to that of Chawla et al.~\cite{chawla2015near}: edges of color $c$ are treated as positive edges ($E^+$), edges of the adversarial color $\gamma$ as negative edges ($E^-$), and all other edges as neutral ($E^\circ$).

For positive and negative edges, the definitions of $e.cost$ and $e.lp$ remain consistent with those in~\cite{chawla2015near}. The other three cases, particularly those involving neutral edges, require more careful treatment.

Consider a negative edge $uv \in E^-$: the LP value is 
\[
e.lp_w(u,v) = \sum_{c' \in L} (1 - x_{uv}^{c'}) \geq 1 - x_{uv}^c.
\]

For a neutral edge $uv \in E^\circ$, the expected cost arises from the event that $u$ and $v$ are not separated by $w$, i.e., at least one of them shares a cluster with $w$. The expected cost is thus given by the probability that $u$ and $v$ are not simultaneously separated from $w$.

The LP contribution in this case is the product of this probability with $x_{uv}^{\phi(uv)}$, where $\phi(uv) \neq c$ is the color of edge $uv$ in the input. While $x_{uv}^{\phi(uv)}$ is not tied to color $c$, we can still bound it below using $x_{uv}^c,\,x_{vw}^c,\,x_{wu}^c$ due to LP constraints~\cite{XiuHTCH22}:
\begin{align*}
x_{uv}^{\phi(uv)} &\geq \max\{x_u^{\phi(uv)}, x_v^{\phi(uv)}\} \tag{\ref{cond:CCC_evdom}} \\
&\geq \max\left\{\frac{1}{2}, 1 - x_u^c, 1 - x_v^c \right\} \tag{\ref{cond:CCC_chroma},\,\ref{cond:CCC_bd}} \\
&\geq \max\left\{\frac{1}{2}, 1 - x_{uv}^c, 1 - x_{vw}^c, 1 - x_{wu}^c \right\}. \tag{\ref{cond:CCC_evdom}}
\end{align*}

Summarizing the results, we express the expected cost and lower bound on the LP value for a fixed pivot $w$ as follows:
\begin{equation}
e.cost_w(u,v) = \begin{cases}
    p_{uw}(1 - p_{vw}) + (1 - p_{uw})p_{vw}, & uv \in E^+; \\
    (1 - p_{uw})(1 - p_{vw}), & uv \in E^-; \\
    1 - p_{uw}p_{vw}, & uv \in E^\circ;
\end{cases}
\end{equation}
\begin{equation}\label{eq:ccc_elp}
e.lp_w(u,v) \geq \begin{cases}
    (1 - p_{uw}p_{vw})x_{uv}^c, & uv \in E^+; \\
    (1 - p_{uw}p_{vw})(1 - x_{uv}^c), & uv \in E^-; \\
    (1 - p_{uw}p_{vw})\max\left\{\frac{1}{2}, 1 - x_{uv}^c, 1 - x_{vw}^c, 1 - x_{wu}^c \right\}, & uv \in E^\circ.
\end{cases}
\end{equation}

These formulations are central to the analysis. Since $\alpha \cdot LP-ALG$ is always at least the expression obtained from the LP lower bound, we can prove that this bound is nonnegative.

As in Section~\ref{subsec:analysis_wCC}, the algorithm achieves an $\alpha$-approximation if the following inequality holds for all triangles $uvw$:
\[
e.cost_w(uv)  + e.cost_u(vw)+ e.cost_v(wu) \leq \alpha \cdot \left( e.lp_w(uv) + e.lp_u(vw)  + e.lp_v(wu)\right).
\]

This inequality leads to the following result on the approximation guarantee for \textsc{LP-CCC}:

\begin{theorem}
\label{thm:CCC_lb}
The approximation factor of \textsc{LP-CCC} for CCC is bigger than $2.11$. 
\end{theorem}

\begin{proof}
    Suppose \textsc{LP-CCC} is an expected $\alpha$-approximation algorithm, where $2 \leq \alpha < 4$.

    Plugging $(\frac{1}{2}, \frac{1}{2}, 1)$ with $(\circ, \circ, -)$, we obtain:
    \begin{equation}
        \alpha \left(1-f^\circ\left(\frac{1}{2}\right)\right)\geq 2\left(1-f^\circ\left(\frac{1}{2}\right)\right)+\left(1-f^\circ\left(\frac{1}{2}\right)\right)^2,
    \end{equation}
    which implies:
    \begin{equation}
    \label{eq:ccc_lb_upper}
        f^\circ\left(\frac{1}{2}\right) \geq 3 - \alpha.
    \end{equation}

    Next, plugging $\left(\frac{1}{4}, \frac{1}{4}, \frac{1}{2}\right)$ with $(+, +, \circ)$ yields:
    \begin{equation}
        \frac{1 - f^+\left(\frac{1}{4}\right)^2 + 2\left(f^+\left(\frac{1}{4}\right) + f^\circ\left(\frac{1}{2}\right) - 2f^+\left(\frac{1}{4}\right)f^\circ\left(\frac{1}{2}\right)\right)}{
        \frac{3}{4}\left(1 - f^+\left(\frac{1}{4}\right)^2\right) + \frac{1}{2}\left(1 - f^+\left(\frac{1}{4}\right)f^\circ\left(\frac{1}{2}\right)\right)
        } \leq \alpha,
    \end{equation}
    which can be rewritten as:
    \begin{equation}
        \left(2 + \left(\frac{\alpha}{2} - 4\right)f^+\left(\frac{1}{4}\right)\right) f^\circ\left(\frac{1}{2}\right) \leq
        \frac{\alpha}{4} \left(5 - 3f^+\left(\frac{1}{4}\right)^2\right) + \left(f^+\left(\frac{1}{4}\right) - 1\right)^2 - 2.
    \end{equation}

    Using the known bound from~\cite{chawla2015near}, plugging $(0, x, x)$ with $(+, +, +)$ yields:
    \begin{equation}
    \label{eq:ccc_lb_lower1}
        f^+(x) \leq 1 - \sqrt{1 - \alpha x},
    \end{equation}
    which implies:
    \begin{align*}
        \left(2 + \left(\frac{\alpha}{2} - 4\right)f^+\left(\frac{1}{4}\right)\right)
        &\geq \left(2 + \left(\frac{\alpha}{2} - 4\right)\left(1 - \sqrt{1 - \frac{\alpha}{4}}\right)\right)\\
        &= 2 \sqrt{1 - \frac{\alpha}{4}} \left(\left(\sqrt{1 - \frac{\alpha}{4}} - \frac{1}{2}\right)^2 + \frac{3}{4}\right) > 0,
    \end{align*}
    for any $\alpha < 4$.

    Therefore, we obtain the following upper bound on $f^\circ\left(\frac{1}{2}\right)$:
    \begin{equation}
    \label{eq:ccc_lb_lower2}
        f^\circ\left(\frac{1}{2}\right) \leq \frac{
        \frac{\alpha}{4}\left(5 - 3f^+\left(\frac{1}{4}\right)^2\right) + \left(f^+\left(\frac{1}{4}\right) - 1\right)^2 - 2
        }{
        2 + \left(\frac{\alpha}{2} - 4\right)f^+\left(\frac{1}{4}\right)
        }.
    \end{equation}

    Define
    \[
        g(x,y):=\frac{
        \frac{y}{4}\left(5 - 3x^2\right) + \left(x - 1\right)^2 - 2
        }{
        2 + \left(\frac{y}{2} - 4\right)x
        }
    \]
    for $x\in \left[0, 1-\sqrt{1-\frac{y}{4}}\right]$ and $y\in [2,4)$, to examine the maximum possible value of $f^\circ\left(\frac{1}{2}\right)$ for a given $\alpha$.

    The partial derivative of $g(x,y)$ with respect to $x$ is:
    \begin{equation}
        \frac{\partial}{\partial x}g(x,y)=\frac{3y-4}{2(8-y)}+\frac{5y^{3}-84y^{2}+368y-448}{2(8-y)(4-8x+xy)^{2}}.
    \end{equation}

    Within the given region, $\frac{\partial}{\partial x}g(x,y)$ is nonzero, as shown in Figure~\ref{fig:CCC_proof_1}. Moreover, evaluating at $x = 0$ gives:
    \begin{align}
        \left.\frac{\partial}{\partial x}g(x,y)\right|_{x=0}=\frac{5 y^3 - 84 y^2 + 416 y - 512}{32(8-y)} = -\frac{5 y^2 - 44 y + 64}{32} \geq \frac{1}{8}.
    \end{align}

    \begin{figure}
        \centering
        \includegraphics[width=0.4\linewidth]{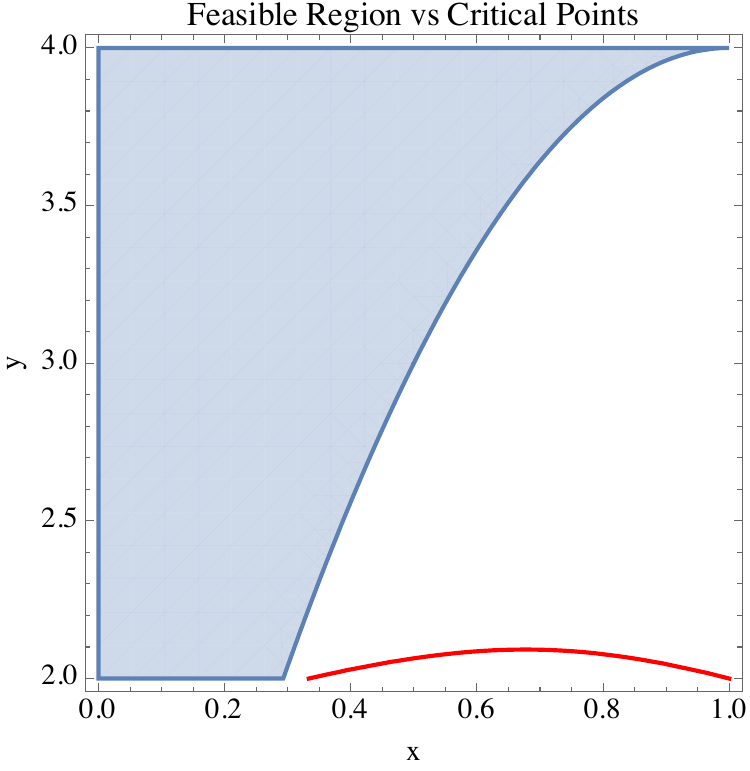}
        \caption{Comparison between the domain of $g$ and the curve where $\frac{\partial}{\partial x}g(x,y)=0.$}
        \label{fig:CCC_proof_1}
    \end{figure}

    Therefore, for any $y_0 \in [2,4)$, the function $g(x,y_0)$ is strictly increasing within the region. Consequently,
    \[
        g(x,y_0)\leq g\left(1-\sqrt{1-\frac{y_0}{4}},y_0\right),
    \]
    which corresponds to the parabolic boundary of the region depicted in Figure~\ref{fig:CCC_proof_1}.

    Parameterizing this curve as $(1-t,4(1-t^2))$ for $t\in\left(0,\frac{1}{\sqrt2}\right]$, the value of $g$ along the curve is:
    \begin{align}
        g(1-t,4(1-t^2))&=\frac{(1-t^2)(2+6t-3t^2)+t^2-2}{2-2(t^2+1)(1-t)}\\
        &=\frac{3 t^3 - 6 t^2 - 4 t + 6}{2(t^2-t+1)}.
    \end{align}
This is the upper bound for $f^\circ\left(\frac{1}{2}\right)$ when $\alpha = 4(1 - t^2)$.

    On the other hand, from~\eqref{eq:ccc_lb_upper}, the lower bound is $3 - \alpha = 3 - 4(1 - t^2) = 4t^2 - 1$.

    Therefore, the inequality
    \begin{equation}
        2(t^2 - t + 1)(4t^2 - 1) \leq 3t^3 - 6t^2 - 4t + 6
    \end{equation}
    must hold.

    Define:
    \[
        h(t) := 2(t^2 - t + 1)(4t^2 - 1) - (3t^3 - 6t^2 - 4t + 6) = 8t^4 - 11t^3 + 12t^2 + 6t - 8.
    \]

    Since $h''(t) = 96t^2 - 66t + 24 > 0$ and $h'(0) = 6 > 0$, $h$ is convex and increasing for $t \geq 0$.

    Finally, since $h\left(\sqrt{1 - \frac{2.11}{4}}\right) > 0$, we conclude that $h(t) > 0$ when $t \geq \sqrt{1 - \frac{2.11}{4}}$, implying that the lower bound for $f^\circ\left(\frac{1}{2}\right)$ exceeds the upper bound when $\alpha \leq 2.11$.

    Therefore, the approximation factor of \textsc{LP-CCC} for CCC must be greater than $2.11$.
\end{proof}

Analogous to the classical CC setting—where the lower bound and the approximation ratio of \textsc{LP-Pivot} differ by less than $0.04$~\cite{chawla2015near}—augmenting the LP rounding with a suitable $f^\circ$ yields the following:

\begin{theorem}
\label{thm:CCC_alg}
\textsc{LP-CCC}, using the rounding functions in~\eqref{eq:round_CCC_pm} and~\eqref{eq:round_CCC_0}, achieves a $2.15$-approximation for Chromatic Correlation Clustering. 
\end{theorem}

\subsection{Proof for Theorem~\ref{thm:CCC_alg}}

    Unlike the pseudometric-weighted case, the variables $x_{vw}^c$ and $x_{wu}^c$ can multiplicatively affect the value of $e.lp_w(u,v)$ in~\eqref{eq:ccc_elp}. Therefore, Lemma~\ref{lem:vd_bd} does not apply, and we must instead verify the inequality $\alpha \cdot LP - ALG \geq 0$ directly.

For sign configurations that do not involve a `$\circ$' (i.e., all edges have $+$ or $-$ signs), the inequality $\alpha \cdot LP - ALG \geq 0$ with $\alpha = 2.06 - \varepsilon$ is already established in~\cite{chawla2015near}. Thus, it suffices to check the remaining $6$ configurations involving at least one `$\circ$' sign.

As before, define $(x, y, z) := (x_{uv}, x_{vw}, x_{wu})$ and $(p_x, p_y, p_z) := (p_{uv}, p_{vw}, p_{wu})$.

\subsubsection{$(\circ,\circ,\circ)$ Triangles}
\label{subsec:CCC000}

\[
2\cdot e.lp_w(u,v) \geq 2(1 - p_{uw}p_{vw})\cdot \frac{1}{2} = (1 - p_{uw}p_{vw}) = e.cost_w(u,v).
\]

Since the same inequality holds for all permutations of the endpoints, we conclude:

\[
2 \cdot LP - ALG \geq 0.
\]

\subsubsection{$(\circ,\circ,+)$ Triangles}
\label{subsec:CCC00+}

If $z \geq \frac{1}{2}$,
\[
2 \cdot e.lp_v(w,u) - e.cost_v(w,u)
\geq (1 - p_x p_y) - (p_x + p_y - 2 p_x p_y)
= (1 - p_x)(1 - p_y) \geq 0,
\]
and by the same reasoning as in Section~\ref{subsec:CCC000}, we also have
\[
2 \cdot e.lp_w(u,v) \geq e.cost_w(u,v), \quad 2 \cdot e.lp_u(v,w) \geq e.cost_u(v,w).
\]
Hence, $2 \cdot LP - ALG \geq 0$.

\medskip

Otherwise, i.e., if $1 - 2z \geq 0$:
\begin{align*}
2 \cdot LP - ALG
&= 2 \cdot \left( (2 - (p_x + p_y)p_z) \max\left\{ \frac{1}{2}, 1 - x, 1 - y, 1 - z \right\} + (1 - p_x p_y)z \right) \\
&\quad - \left( 1 - p_y p_z + 1 - p_x p_z + p_x + p_y - 2 p_x p_y \right) \\
&\geq 2 \cdot \left( (2 - (p_x + p_y)p_z)(1 - z) + (1 - p_x p_y)z \right) \\
&\quad - \left( 2 - (p_x + p_y)p_z + (1 - p_x p_y) - (1 - p_x)(1 - p_y) \right) \\
&= (1 - 2z)(2 - (p_x + p_y)p_z) + (1 - 2z)(p_x p_y - 1) + (1 - p_x)(1 - p_y) \\
&= (1 - 2z)\left( 2 - (p_x + p_y)p_z + p_x p_y - 1 \right) + (1 - p_x)(1 - p_y) \\
&\geq (1 - 2z)\left( 2 - p_x - p_y + p_x p_y - 1 \right) + (1 - p_x)(1 - p_y) \\
&= (1 - 2z)(1 - p_x)(1 - p_y) + (1 - p_x)(1 - p_y) \geq 0.
\end{align*}

    \subsubsection{$(\circ,\circ,-)$ Triangles}
    \label{subsec:CCC00-}
  If $z \leq 1 - \frac{1}{\alpha}$, then as in Section~\ref{subsec:CCC00+}, we have $\alpha \cdot LP - ALG \geq 0$.

    Otherwise, since $1-z\leq 1/\alpha<1/2$,
    \begin{align*}
    &\alpha\cdot LP-ALG\\
    &=\alpha\cdot\left((2-(p_x+p_y)z)\max\left\{\frac{1}{2},1-x,1-y,1-z\right\}+(1-p_xp_y)(1-z)\right)\\
    &\quad -\left(1-p_yz+1-p_xz+(1-p_x)(1-p_y)\right)\\
    &=\alpha(2-(p_x+p_y)z)\max\left\{\frac{1}{2},1-x,1-y\right\}
    +\alpha(1-p_xp_y)(1-z)\\
    &\quad-\left(2-(p_x+p_y)z+(1-p_x)(1-p_y)\right).
    \end{align*}
  Fixing the values of $x$ and $y$, the function is affine with respect to $z$, and its slope is given by
\begin{align*}
    &-\alpha(p_x + p_y) \max\left\{\frac{1}{2}, 1 - x, 1 - y \right\} - \alpha(1 - p_x p_y) + (p_x + p_y) \\
    &\leq \left(-\frac{\alpha}{2} + 1\right)(p_x + p_y) - \alpha(1 - p_x p_y) \leq 0.
\end{align*}
Therefore, the function is minimized when $z$ is maximized.

Now, \emph{w.l.o.g.}, assume $x \leq y$, and consider the following three cases:

    \begin{enumerate}
        \item $x\geq 1/2$

        $p_x,p_y\geq 0.85$ as $x,y\geq 1/2$. Applying $z=1$,
        \begin{align*}
            &\alpha\cdot LP-ALG\\
            &\geq \alpha(2-p_x-p_y)\max\left\{\frac{1}{2},1-x,1-y\right\}-(2-p_x-p_y)-(1-p_x)(1-p_y)\\
            &= \left(\frac{\alpha}{2}-1\right)(2-p_x-p_y)-(1-p_x)(1-p_y)\\
            &\geq \left(\frac{\alpha}{2}-1\right)(2-p_x-p_y)-\frac{1}{4}(2-p_x-p_y)^2\\
            &=\frac{1}{4}(2-(p_x+p_y))(2\alpha -6+(p_x+p_y)).
        \end{align*}
        Since $6-2\alpha=1.7\leq p_x+p_y\leq 2$, $\alpha\cdot LP-ALG\geq 0$.
        \item $x<1/2,\,x+y\geq 1$
        
        Applying $z=1$ again,
        \begin{align*}
            &\alpha\cdot LP-ALG\\
            &\geq \alpha(2-p_x-p_y)\max\left\{\frac{1}{2},1-x,1-y\right\}-(2-p_x-p_y)-(1-p_x)(1-p_y)\\
            &= (\alpha(1-x)-1)(2-p_x-p_y)-(1-p_x)(1-p_y).
        \end{align*}
        Note that the function is affine with respect to $p_y$; therefore, the function is minimized when $p_y$ reaches the boundary, i.e., $y=1$ or $y=1-x$.

        If $y=1$, then $p_y=1$, hence
        \begin{align*}
            &(\alpha(1-x)-1)(2-p_x-p_y)-(1-p_x)(1-p_y)\\
            &=(\alpha(1-x)-1)(1-p_x)\\
            &\geq \left(\frac{\alpha}{2}-1\right)(1-p_x)\geq 0.
        \end{align*}

        On the other hand, if $y=1-x$, then $p_y=\frac{3}{10}y+\frac{7}{10}=1-\frac{3}{10}x$ as $y>1/2$, hence
        \begin{align*}
            &(\alpha(1-x)-1)(2-p_x-p_y)-(1-p_x)(1-p_y)\\
            &=(\alpha(1-x)-1)\left(1-\frac{17}{10}x+\frac{3}{10}x\right)-\frac{3}{10}\left(1-\frac{17}{10}x\right)x\\
            &=3.52\cdot \left(x-\frac{1}{2}\right)\left(x-\frac{115}{176}\right)\geq 0
        \end{align*}
       
        as $x< 1/2$ and $p_x=\frac{17}{10}x$.

        \item $x+y<1$

        Applying $z=x+y$,
        \begin{align*}
            &\alpha\cdot LP-ALG\\
            &=\alpha(2-(p_x+p_y)(x+y))\max\left\{\frac{1}{2},1-x,1-y\right\}
    +\alpha(1-p_xp_y)(1-x-y)\\
    &\quad -\left(2-(p_x+p_y)(x+y)+(1-p_x)(1-p_y)\right)\\
    &=(\alpha(1-x)-1)(2-(p_x+p_y)(x+y))+\alpha(1-p_xp_y)(1-x-y)-(1-p_x)(1-p_y)\\
    &\geq 0.3\cdot (\alpha(1-x)-1)+\alpha(1-p_xp_y)(1-x-y)-(1-p_x)(1-p_y).
        \end{align*}
        For the last inequality, since $p_x\leq 1.7x$ and $p_y\leq 1.7y$, $2-(p_x+p_y)(x+y)\geq 2-1.7(x+y)^2\geq0.3$.

        Differentiating the function by $y$,
        \begin{align*}
            &-\alpha p_x\frac{dp_y}{dy}\cdot (1-x-y)-\alpha(1-p_xp_y)+\frac{dp_y}{dy}\cdot (1-p_x)\\
            &\leq -\alpha p_x\frac{dp_y}{dy}\cdot (1-x-y)-\alpha(1-p_x)+\frac{dp_y}{dy}\cdot (1-p_x)<0
        \end{align*}
        since $\frac{dp_y}{dy}\leq 1.7<\alpha$.
        Hence, the function is minimized when $y$ is maximized.

        Applying $y=1-x$, since $y=1-x>1/2$, $p_y=\frac{3}{10}y+\frac{7}{10}=1-\frac{3}{10}x$ and thus
        \begin{align*}
            & 0.3\cdot (\alpha(1-x)-1)+\alpha(1-p_xp_y)(1-x-y)-(1-p_x)(1-p_y)\\
            &=0.3\cdot (\alpha(1-x)-1)-\left(1-\frac{17}{10}x\right)\cdot \frac{3}{10}x\\
            &=0.51\left(x-\frac{1}{2}\right)\left(x-\frac{23}{17}\right)\geq 0
        \end{align*}
        since $x<1/2$.
    \end{enumerate}
   
    \subsubsection{$(\circ,+,+)$ Triangles}
\label{subsec:CCC0++}

For Sections~\ref{subsec:CCC0++} and~\ref{subsec:CCC0+-}, we utilize the following lemma, which serves as a weaker version of Lemma~\ref{lem:vd_bd}:

\begin{lemma}
\label{lem:ccc_bd}
Let the sign configuration be such that the edge $uv$ is the only edge labeled with a `$\circ$' among the three edges of the triangle. Suppose the rounding function $f^\circ$ is piecewise affine. Then $\mathcal{C} \geq 0$ for all possible configurations of $(x_{uv}, x_{vw}, x_{wu})$ if the inequality holds in each of the following cases:
\begin{enumerate}
    \item The triangle inequality is tight for $(x_{uv}, x_{vw}, x_{wu})$;
    \item $x_{uv}$ lies at an endpoint of a domain interval for some piece of the function $f^\circ$;
    \item $x_{uv} = \min\left\{ \frac{1}{2}, x_{vw}, x_{wu} \right\}$.
\end{enumerate}
\end{lemma}
 
    \begin{proof}
Fix the values of $x_{vw}$ and $x_{wu}$, and treat $\mathcal{C}$ as a function of $x_{uv}$ and $p_{uv}$. Within this function, the term involving $x_{uv}$ appears only in 
\[
e.lp_w(uv) = (1 - p_{vw}p_{wu}) \cdot \max\left\{ \frac{1}{2}, 1 - x_{uv}, 1 - x_{vw}, 1 - x_{wu} \right\},
\]
which is piecewise affine with two domain intervals: 
\[
\left[0, \min\left\{ \frac{1}{2}, x_{vw}, x_{wu} \right\} \right] \quad \text{and} \quad \left[ \min\left\{ \frac{1}{2}, x_{vw}, x_{wu} \right\}, 1 \right].
\]
The remaining terms in $\mathcal{C}$ are affine in $p_{uv}$.

Hence, $\mathcal{C}$ is a piecewise affine function defined over two domains in the $(x_{uv}, p_{uv})$ space:
\[
\left[ 0, \min\left\{ \frac{1}{2}, x_{vw}, x_{wu} \right\} \right] \times [0,1]
\quad \text{and} \quad 
\left[ \min\left\{ \frac{1}{2}, x_{vw}, x_{wu} \right\}, 1 \right] \times [0,1].
\]
Since $f^\circ$ is piecewise affine, substituting $p_{uv} = f^\circ(x_{uv})$ yields a piecewise affine function in $x_{uv}$. The breakpoints of this composition occur at:
\begin{enumerate}
    \item the endpoints of the domain intervals of $f^\circ$, and
    \item the point $x_{uv} = \min\left\{ \frac{1}{2}, x_{vw}, x_{wu} \right\}$ where the domain of the piecewise max function changes.
\end{enumerate}

Therefore, the minimum of $\mathcal{C}$ occurs either at domain boundaries or at breakpoints—precisely the three cases specified in the lemma.
\end{proof}

    Moreover, in this particular case, an even stronger argument can be made: Examining the full formula of $\mathcal{C}=\alpha\cdot LP-ALG$ with this sign configuration,
    \begin{align*}
    &\alpha\cdot LP-ALG\\
    &=\alpha\cdot\left((1-p_yp_z)\max\left\{\frac{1}{2},1-x,1-y,1-z\right\}+(1-p_xp_z)y+(1-p_xp_y)z\right)\\
    &\quad-\left(1-p_yp_z+p_x+p_z-2p_xp_z+p_x+p_y-2p_xp_y\right)\\
    &=\alpha(1-p_yp_z)\max\left\{\frac{1}{2},1-x,1-y,1-z\right\}+\left(-\alpha(yp_z+p_yz)-2+2(p_y+p_z)\right)p_x\\
    &\quad+\alpha(y+z)-(1+p_y+p_z-p_yp_z)
    \end{align*}
    The first term is a monotonically decreasing function of $x$; the second term is a monotonically decreasing function of $p_x$ since $-\alpha(yp_z+p_yz)-2+2(p_y+p_z)\leq0$, as shown in the following Lemma:
    \begin{lemma}
        For $f^+$ given as equation~\ref{eq:round_CCC_pm} and $\alpha\in \left[\frac{1}{0.5095},\frac{2}{3\cdot 0.5095-0.19}\right]\approx [1.963,2.988]$,
        \[\alpha(xf^+(y)+f^+(x)y)+2-2(f^+(x)+f^+(y))\geq 0\]
        in $[0,1]^2$.
    \end{lemma}
    \begin{proof}
        \emph{w.l.o.g.}\ $x\leq y$.
        Consider the $2$ possible cases, based on the value of $y$.
        \begin{enumerate}
            \item $y\geq0.5095$

            In this case, $f^+(y)=1$ and thus
            \begin{align*}
                &\alpha(xf^+(y)+f^+(x)y)+2-2(f^+(x)+f^+(y))\\
                &=\alpha(x+f^+(x)y)-2f^+(x),
            \end{align*}
            which is an increasing function of $y$.
            
            If $x\geq0.5095$, setting $y=x$,
            \[\alpha(x+f^+(x)y)-2f^+(x)\geq 2\alpha x-2\geq0\]
            since $x>1/2$.
            Otherwise, setting $y=0.5095$,
            \[\alpha(x+f^+(x)y)-2f^+(x)\geq (\alpha\cdot 0.5095-2)f^+(x)+\alpha x\]
            is concave of $x$ since $\alpha\cdot 0.5095-2<0$ and $f^+$ is convex in the range.
            Hence, this is minimized at either $x=0$ or $x=0.5095$, where the former case results in $0$ and the latter case results in $2\alpha\cdot 0.5095-2\geq 0$.
            \item $y< 0.5095$

            As a function of $y$,
            \begin{align*}
                &\alpha(xf^+(y)+f^+(x)y)+2-2(f^+(x)+f^+(y))\\
                &=(\alpha x-2)f^+(y)+\alpha f^+(x)y+2-2f^+(x)
            \end{align*}
            is concave since $\alpha x-2\leq \alpha\cdot 0.5095-2<0$ and $f^+$ is convex in the range.
            Hence, this is minimized at either $y=x$ or $y=0.5095$, where the latter case reduce to Case 1, and the former case results in
            \begin{align*}
                &(\alpha x-2)f^+(y)+\alpha f^+(x)y+2-2f^+(x)\\
                &=2\left((\alpha x -2)f^+(x)+1\right),
            \end{align*}
            which is $2\geq 0$ if $x<0.19$, or
            \[
                2\left((\alpha x -2)\left(\frac{x-0.19}{0.5095-0.19}\right)^2+1\right)
            \]
            otherwise.
            This is a decreasing function when $x\in\left[0,\frac{0.19+4/\alpha}{3}\right]$, which contains $[0,0.5095]$.
            Therefore, this is minimized when $x=0.5095$, which also reduce to Case 1.
        \end{enumerate}
        This completes the proof of the lemma.
    \end{proof}

    Therefore, as $p_x$ is a monotonically increasing function of $x$, the function is minimized when $x$ is maximized.

    Now, \emph{w.l.o.g.}\ $y\leq z$ and consider all $2$ possible cases:
    \begin{enumerate}
        \item $y+z\geq 1$

        Applying $x=1$,
        \begin{align*}
            &\alpha\cdot LP-ALG\\
            &\geq \alpha\cdot\left((1-p_yp_z)\max\left\{\frac{1}{2},1-y,1-z\right\}+(1-p_y)z+(1-p_z)y\right)\\
    &\quad-(3-p_yp_z-p_z-p_y)\\
    &\geq \alpha\cdot\left(\frac{1}{2}(1-p_yp_z)+(1-p_z)y+(1-p_y)z\right)-(3-p_yp_z-p_z-p_y)\\
    &\geq \alpha\cdot\left(\frac{1}{2}(1-p_yp_z)+(1-p_y)\frac{p_z}{2}+(1-p_z)\frac{p_y}{2}\right)-(3-p_yp_z-p_z-p_y)\\
    &=2.225\left(\left(\frac{6}{89}\right)^2-\left(p_y-\frac{83}{89}\right)\left(p_z-\frac{83}{89}\right)\right).
        \end{align*}
        Since $z\geq 1/2$, $p_z>0.94>\frac{83}{89}$, thus the function is decreasing with respect to $p_y$. The minimum value of the function is therefore $0$ when $p_y=p_z=1$.
    \item $y+z<1$
    
    We use the similar argument with of those conducted by Chawla et al.\cite{chawla2015near} on the $(+,+,-)$ case.
    For simplicity, let $a=0.19,\,b=0.5095$.

    If $y<a$, $p_y=0$, hence
    \begin{align*}
        &\alpha \cdot LP-ALG\\
        &=\alpha\cdot\left((1-y)+(1-p_xp_z)y+z\right)-\left(1+2p_x+p_z-2p_xp_z\right)\\
    &=-\alpha p_xp_z y +\alpha(1+z)-\left(1+2p_x+p_z-2p_xp_z\right).
    \end{align*}
    This is a non-increasing function of $y$, and thus it is sufficient to check on $y=a$.
    
    If $z>b$, $p_z=1$, hence
    \begin{align*}
        &\alpha \cdot LP-ALG\\
        &=\alpha\cdot\left((1-p_y)(1-y)+(1-p_x)y+(1-p_xp_y)z\right)-\left(2-2p_xp_y\right)\\
        &=\alpha(1-p_x p_y)z+\alpha ((1-p_y)(1-y)+(1-p_x)y)-2(1-p_xp_y).
    \end{align*}
    This is a non-decreasing function of $z$, and thus it is sufficient to check on $z=b$.

    Now the problem reduces to only $a\leq y\leq z\leq b$ case, where the formula for $p_y$ and $p_z$ are both quadratic.
    
    We now show that the function is minimized when $y=z$ as well as $x=y+z$.
    Consider the parametrization $(x,y,z)=(x,t-c,t+c)$, where $t\in[a,1/2],\,c\in[0,\min\{t-a,b-t\}],$ and $x=2t$ since $\alpha\cdot LP-ALG$ is a monotonically increasing function of $x$.
    Then
    \begin{align*}
        &\alpha \cdot LP-ALG\\
        &=\alpha\cdot\left((1-p_yp_z)(1-t)+(1-p_xp_z)(t-c)+(1-p_xp_y)(t+c)\right)\\
    &\quad-\left(1+2p_x+p_y+p_z-p_yp_z-2p_xp_z-2p_xp_y\right)\\
    &\quad+\alpha(1-p_y)c.
    \end{align*}
    Fixing the value of $t$, the sum of first $2$ terms is a negative even quartic function of $c$, where the negativity is obtained from $-\alpha(1-t)+1\leq -\alpha(1-b)+1<0$.
    Thus, we can write the function as
    \[P(t)c^4+Q(t)c^2+\alpha(1-p_y)c+R(t),\]
    where $P(t)=\frac{-\alpha(1-t)+1}{(b-a)^4}<0$.
    Since $\alpha(1-p_y)c\geq0$ if $c\geq 0$, it is sufficient to show that $Q\geq 0$ and $\sqrt{\frac{Q}{-P}}\geq \min\{t-a,b-t\}$ to prove that the function is minimized when $c=0$.

    Calculating the $Q$, this results in
    \begin{align*}
        Q(t)
        &=\frac{2}{(b-a)^2}\cdot\left((\alpha(1-t)-1)\left(\frac{t-a}{b-a}\right)^2+\alpha p_x(t-2a)-(1-2p_x)\right)\\
        &=-P(t)\cdot (t-a)^2+\frac{2}{(b-a)^2}\cdot\left(-1+(\alpha t +2-2\alpha a)p_x\right)\\
        &\geq -P(t)\cdot (t-a)^2,
    \end{align*}
    where the final inequality is derived from
    \[\frac{1}{\alpha t+2-2\alpha a}\leq \frac{1}{2-\alpha a}\approx 0.628<0.646=f^\circ (2\cdot 0.19)\leq p_x.\]
    Hence, $Q\geq -P(t)\cdot (t-a)^2>0$ and therefore $\sqrt{\frac{Q}{-P}}\geq t-a\geq \min\{t-a,b-t\}$, which implies that $\alpha \cdot LP-ALG$ is minimized at $c=0$, i.e.\ at $(2t,t,t)$.

    Finally, plugging $(x,y,z)=(2t,t,t)$,
    \begin{align*}
        &\alpha \cdot LP-ALG\\
        &=\alpha\cdot\left((1-p_y^2)(1-y)+2(1-p_xp_y)y\right)-\left(1+2p_x+2p_y-p_y^2-4p_xp_y\right)\\
        &=\left(\alpha(1-p_y^2)(1-t)+2\alpha t-1-2p_y+p_y^2 \right)-(2\alpha p_y t+2-4p_y)p_x.
    \end{align*}

    The top left region of Figure~\ref{fig:round_CCC_region} shows the region $(2t,f^\circ(2t))$ where the formula above is negative when $f^+$ is given as equation~\ref{eq:round_CCC_pm}.
    Since our proposed $f^\circ$ doesn't intersect the region, $\alpha \cdot LP-ALG\geq 0$ if $y=z$, and thus $\alpha \cdot LP-ALG\geq 0$ whenever $a\leq y\leq z\leq b$.
    \end{enumerate}

    \subsubsection{$(\circ,+,-)$ Triangles}
    \label{subsec:CCC0+-}
    We would prove that the lower bound $\mathcal{C}'$ of $\mathcal{C}$, which is defined as
    \begin{align*}
        &\alpha\cdot LP-ALG\\
        &=\alpha\cdot\left((1-p_yz)\max\left\{\frac{1}{2},1-x,1-y,1-z\right\}+(1-p_xp_y)(1-z)+(1-p_xz)y\right)\\
        &\quad-\left(1-p_yz+p_x+z-2p_xz+(1-p_x)(1-p_y)\right)\\
        &\geq \alpha\cdot\left((1-p_yz)\max\left\{\frac{1}{2},1-y,1-z\right\}+(1-p_xp_y)(1-z)+(1-p_xz)y\right)\\
        &\quad-\left(1-p_yz+p_x+z-2p_xz+(1-p_x)(1-p_y)\right)\triangleq\mathcal{C}'
    \end{align*}
    is nonnegative.
    
    Since this is now affine with respect to $p_x$, this further enhances the argument Lemma~\ref{lem:ccc_bd} to check only on the endpoint.
    % \[\alpha\cdot\left((1-p_yz)\max\left\{\frac{1}{2},1-x,1-y,1-z\right\}+(1-p_xp_y)(1-z)+(1-p_xz)y\right)-\left(1-p_yz+p_x+z-2p_xz+(1-p_x)(1-p_y)\right)\]
    \begin{enumerate}
        \item $y\geq 1/2,\,z\leq 1/2$

        As in \ref{subsec:CCC00+} and \ref{subsec:CCC00-}, $2\cdot LP-ALG\geq 0$.

        \item $y,z\geq 1/2$

        Similar with \ref{subsec:CCC0++}, we can deduce a stronger argument than of Lemma~\ref{lem:ccc_bd}: Examining the full formula of $\mathcal{C}'$,
        \begin{align*}
            \mathcal{C}'&=\alpha\cdot\left((1-p_yz)\max\left\{\frac{1}{2},1-y,1-z\right\}+(1-p_xp_y)(1-z)+(1-p_xz)y\right)\\
            &\quad-\left(1-p_yz+p_x+z-2p_xz+(1-p_x)(1-p_y)\right)\\
            &=(-\alpha yz -\alpha p_y(1-z)-p_y+2z)p_x\\
            &\quad +\alpha(1-z)+\alpha y+p_y z-z+p_y-2+\alpha (1-p_yz)\max\left\{\frac{1}{2},1-y,1-z\right\}.
        \end{align*}
        The first term is a monotonically decreasing function of $p_x$ since
        \[-\alpha yz -\alpha p_y(1-z)-p_y+2z=(\alpha (p_y-y)+2)z-(\alpha+1)p_y\]
        is an increasing function of $z$ since $p_y\geq y$ in the range,
        \[-\alpha yz -\alpha p_y(1-z)-p_y+2z\leq -\alpha y-p_y+2\leq -\frac{\alpha}{2}-f^+\left(\frac{1}{2}\right)+2\approx -0.016<0.\]

        Hence, the function is minimized when $x$ is maximized.
        Setting $x=1$,
        \begin{align*}
            \mathcal{C}'&\geq\alpha\cdot\left((1-p_yz)\cdot \frac{1}{2}+(1-p_y)(1-z)+(1-z)y\right)-(2-p_yz-z)\\
            &=\left(-\alpha+1-\alpha y+\left(\frac{\alpha}{2}+1\right)p_y\right)z+\alpha y-\alpha p_y +\frac{3}{2}\alpha -2.
        \end{align*}
        Since $y\geq1/2$ and $p_y\leq 1$, 
        \[-\alpha+1-\alpha y+\left(\frac{\alpha}{2}+1\right)p_y\leq -\alpha+2<0.\]
        Hence, the function is minimized when $z$ is maximized.
        Setting $z=1$,
        \begin{align*}
            \mathcal{C}'
            &\geq \frac{\alpha}{2}-1+\left(-\frac{\alpha}{2}+1\right)p_y\\
            &=\left(\frac{\alpha}{2}-1\right)(1-p_y)\geq 0.
        \end{align*}
        \item $y\leq z,\,y<1/2$
        \begin{align*}
            \mathcal{C}'
            &=\alpha\cdot\left((1-p_yz)(1-y)+(1-p_xp_y)(1-z)+(1-p_xz)y\right)\\
            &\quad-\left(1-p_yz+p_x+z-2p_xz+(1-p_x)(1-p_y)\right)\\
            &=(-\alpha(1+p_y-p_yy-p_xp_y+p_xy)+2p_x+p_y-1)z\\
            &\quad+\alpha(2-p_xp_y)-2+p_y-p_xp_y.
        \end{align*}
        This is a decreasing function of $z$ since
        \begin{align*}
            &-\alpha(1+p_y-p_yy-p_xp_y+p_xy)+2p_x+p_y-1\\
            &=(2+\alpha p_y - \alpha y)p_x-\alpha(1+p_y-p_y y)+p_y-1
        \end{align*}
        is an increasing function of $p_x$ since $2+\alpha p_y - \alpha y\geq 2-\frac{\alpha}{2}\geq0$,
        \begin{align*}
            &-\alpha(1+p_y-p_yy-p_xp_y+p_xy)+2p_x+p_y-1\\
            &\leq -\alpha(1-p_y y+y)+1+p_y\\
            &= -(1-p_y)(1+\alpha y)-\alpha+2\leq -\alpha+2<0.
        \end{align*}
        Hence, the function is minimized when $z$ is maximized.
        Setting $z=1$,
        \begin{align*}
            \mathcal{C}'&\geq \alpha\cdot\left((1-p_y)(1-y)+(1-p_x)y\right)-\left(3-2p_x-2p_y+p_x p_y\right).
        \end{align*}
        As mentioned above, it is sufficient to check on $x=1$ and $x=1-y$.
        \begin{enumerate}
            \item $x=1$
            \begin{align*}
                 \mathcal{C}'&= \alpha\cdot(1-p_y)(1-y)-(1-p_y)\\
                 &=(1-p_y)(\alpha-1-\alpha y)\\
                 &\geq (1-p_y)\left(\frac{\alpha}{2}-1\right)\geq 0.
            \end{align*}
            \item $x=1-y$

            Since $x=1-y>1/2$, $p_x=\frac{3}{10}x+\frac{7}{10}=1-\frac{3}{10}y$.
            Therefore,
            \begin{align*}
                \mathcal{C}'&= \alpha\cdot\left((1-p_y)(1-y)+\frac{3}{10}y^2\right)-\left(1+\frac{3}{10}y\right)(2-p_y)+1.
            \end{align*}
            If $y<0.19$,
            \begin{align*}
                \mathcal{C}'&= \alpha\cdot\left(1-y+\frac{3}{10}y^2\right)-2\left(1+\frac{3}{10}y\right)+1\\
                &=\frac{1}{200}(129y^{2}-550y+230)>0.65\geq 0.
            \end{align*}
            Otherwise, i.e.\ $0.19\leq y<1/2<0.5095$,
            \[
                \mathcal{C}'\approx 0.74331 + 2.39737 y - 19.7409 y^2 + 24.0007 y^3,
            \]
            whose minimum value in the region is approximately $7.236\times 10^{-6}\geq 0$ in $y\approx 0.4788$.
        \end{enumerate}
        
        \item $z< y< 1/2$
        \begin{align*}
            \mathcal{C}'
            &=\alpha\cdot\left((1-p_yz)(1-z)+(1-p_xp_y)(1-z)+(1-p_xz)y\right)\\
            &\quad-\left(1-p_yz+p_x+z-2p_xz+(1-p_x)(1-p_y)\right)\\
            &=\alpha p_y z^2+(-\alpha(2+p_y-p_xp_y+p_xy)+2p_x+p_y-1)z\\
            &\quad+\alpha(2+y-p_xp_y)-2+p_y-p_xp_y.
        \end{align*}
        Differentiating by $z$,
        \begin{align*}
            &2\alpha p_yz-\alpha(2+p_y-p_xp_y+p_xy)+2p_x+p_y-1\\
            &\leq \alpha p_y-\alpha(2+p_y-p_xp_y+p_xy)+2p_x+p_y-1\\
            &= (-2\alpha +2p_x+p_y)-\alpha((1-p_x)p_y+p_xy)-1\leq 0
        \end{align*}
        since $z<1/2$ and $-2\alpha +2p_x+p_y\leq -2\alpha+3=-1.3<0$.
        Therefore, the function is minimized when $z$ is maximized.
        
        Plugging $z=y$, this now reduce to Case 3 with $z=y$ due to the continuity of $\mathcal{C}'$.
        
    \end{enumerate}

    \subsubsection{$(\circ,-,-)$ Triangles}
    \label{subsec:CCC0--}
    
    As in \ref{subsec:CCC000}, $2e.lp_w(u,v)\geq e.cost_w(u,v)$.
    For other two terms,
    \begin{align*}
        e.lp_u(v,w)+e.lp_v(w,u)&\geq(1-p_{uv}z)(1-y)+(1-p_{uv}y)(1-z)\\
        &\geq (1-p_{uv})(1-y)+(1-p_{uv})(1-z)\\
        &=(1-p_{uv})(1-p_{vw})+(1-p_{uv})(1-p_{wu})\\
        &=e.cost_v(w,u)+e.cost_u(v,w).
    \end{align*}
    Hence, $2\cdot LP-ALG\geq 0$.
\\
\\
This completes our case analysis.

\section{Conclusion}\label{sec:con}

In this work, we studied two important variants of correlation clustering: \emph{pseudometric-weighted correlation clustering} and \emph{chromatic correlation clustering }. For both problems, we developed and analyzed specialized rounding functions that are essential for achieving improved approximation guarantees via the \textsc{LP-Pivot} algorithm.

For the pseudometric-weighted setting, we proposed a piecewise-linear rounding function tailored for the setting that achieves a $10/3$-approximation, and showed that no alternative function within our analytical framework can yield a better factor.
For the chromatic correlation clustering  variant, we designed a distinct rounding function  that respects the constraints imposed by color restrictions and achieves an approximation factor of $2.15$. The function is constructed using a piecewise-linear form and leverages a careful analysis of triple costs.

Overall, our work highlights the importance of designing principled and variant-specific rounding strategies to extend LP-based techniques to structured clustering problems, yielding  strong theoretical guarantees.

%\section*{Acknowledgment}
%The research in this work by Chenglin Fan and Dahoon Lee has been partially supported by the SNU Startup Fund.

\bibliographystyle{plain}
\bibliography{reference}

\newpage
\appendix

%\textbf{Technical Appendices and Supplementary Material}
\appendix
%\section{Proof of Rounding Functions}

\section{Proof for Theorem~\ref{thm:wCC_alg}}
We can apply the same argument with Appendix~\ref{sec:wCC_old}, with the partition by regions (Figure~\ref{fig:wCC_new_region}) and formula for $e.cost$ and $e.lp$ (Table~\ref{tab:wCC_new}) changed accordingly.
\begin{figure}
    \centering
    \includegraphics[width=0.40\linewidth]{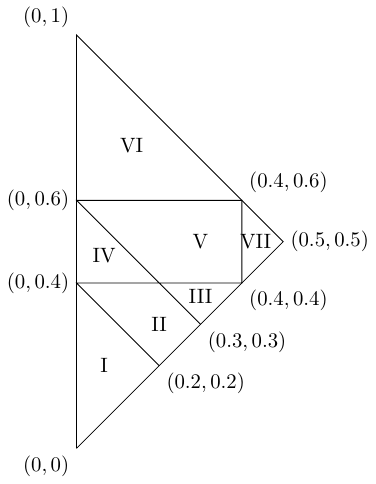}
    \caption{Configuration of regions induced from~\eqref{eq:round_wCC}.}
    \label{fig:wCC_new_region}
\end{figure}
\begin{table}
% \centering
\caption{Formula for $e.cost$ and $e.lp$ for each configuration.}
\label{tab:wCC_new}

\begin{tabular}{llrrrr}
\toprule
 &  & \multicolumn{2}{c}{$x_{uv}=x$} & \multicolumn{2}{c}{$x_{vw}=y$} \\
\cmidrule(r){3-4}\cmidrule(r){5-6}
Region       & Sign     & $e.cost$        & $e.lp$       & $e.cost$        & $e.lp$            \\
       \midrule
I      & `$+$'  & $0$ &  $x$   &    $0$ &  $y$     \\
       & `$-$'  & $1$ & $1-x$ &   $1$  &  $1-y$    \\
II     & `$+$'  &   $\frac{5}{3}z$     &$x$   &  $\frac{5}{3}z$   &$y$\\
       & `$-$'  &$1-\frac{5}{3}z$  &$1-x$&   $1-\frac{5}{3}z$ &$1-y$\\
III    & `$+$'  &   $1$     &$x$   &  $1$   &$y$\\
       & `$-$'  &$0$  &$1-x$&   $0$ &$1-y$\\
IV    & `$+$'  &$\frac{5}{3}y+\frac{5}{3}z-\frac{50}{9}yz$ &  $(1-\frac{25}{9}yz)x$  &  $\frac{5}{3}z$ &  $y$\\
       & `$-$'  &$(1-\frac{5}{3}y)(1-\frac{5}{3}z)$  &  $(1-\frac{25}{9}yz)(1-x)$ & $1-\frac{5}{3}z$&$1-y$\\
V      & `$+$'  &     $1-\frac{5}{3}y$&  $(1-\frac{5}{3}y)x$&  $1$&   $y$\\
       & `$-$'  &$0$&  $(1-\frac{5}{3}y)(1-x)$&    $0$&  $1-y$\\
VI     & `$+$'  &$0$& $0$&$1$&$y$\\
       & `$-$'  &$0$& $0$&$0$&$1-y$\\ 
VII     & `$+$'  &$1-\frac{5}{3}y$&  $(1-\frac{5}{3}y)x$& $1-\frac{5}{3}x$&$(1-\frac{5}{3}x)y$\\
       & `$-$'  &$0$&   $(1-\frac{5}{3}y)(1-x)$&   $0$&  $(1-\frac{5}{3}x)(1-y)$\\
       \bottomrule
\end{tabular}
\begin{tabular}{llrr}
\toprule
 &   & \multicolumn{2}{c}{$x_{wu}=z=x+y$} \\
\cmidrule(r){3-4}
Region & Sign & $e.cost$ & $e.lp$ \\
       \midrule
I      & `$+$'  &  $0$  & $z$   \\
       & `$-$'  & $1$ & $1-z$  \\
II     & `$+$'  &$0$&$z$\\
       & `$-$'  &$1$&$1-z$\\
III    & `$+$'  &$0$&$z$\\
       & `$-$'  &$1$&$1-z$\\
IV    & `$+$'  &$\frac{5}{3}y$&  $z$\\
       & `$-$'  & $1-\frac{5}{3}y$&  $1-z$\\
V      & `$+$'  & $\frac{5}{3}y$&  $z$\\
       & `$-$'    $1-\frac{5}{3}y$&$1-z$\\
VI     & `$+$'  &$1$& $z$\\
       & `$-$'  &   $0$&  $1-z$\\ 
VII     & `$+$' & $\frac{5}{3}x+\frac{5}{3}y-\frac{50}{9}xy$&    $(1-\frac{25}{9}xy)z$\\
       & `$-$'  &    $(1-\frac{5}{3}x)(1-\frac{5}{3}y)$&    $(1-\frac{25}{9}xy)(1-z)$\\
       \bottomrule
\end{tabular}
\end{table}

Define $(x,y,z):=(x_{uv},x_{vw},x_{wu})$ and \emph{w.l.o.g.}\ $x\leq y\leq z$.

\subsection{Triangle inequality is tight.}
\subsubsection{Region I, II, III, VI}
Affine formulae; refer to Table~\ref{tab:wCC_new_extremal}.
Also note that the factor $10/3$ is tight for $(0,\frac{2}{5},\frac{2}{5})$ with $(+,+,+)$, 

\begin{table}[]
\caption{Value of $\frac{10}{3}\cdot e.lp-e.cost$ for each extremal point of the region.}
    \label{tab:wCC_new_extremal}
    \centering
\begin{tabular}{lllrrr}
\toprule
Region&$(x,y)$&Sign& $\frac{10}{3}\cdot e.lp_w-e.cost_w$& $\frac{10}{3}\cdot e.lp_u-e.cost_u$&$\frac{10}{3}\cdot e.lp_v-e.cost_v$\\
\midrule
I&$(0,0)$&`$+$'&$0$&$0$&$0$\\
&&`$-$'&$7/3$&$7/3$&$7/3$\\
&$(0,\frac{2}{5})$&`$+$'&$0$&$4/3$&$4/3$\\
&&`$-$'&$7/3$&$1$&$1$\\
&$(\frac{1}{5},\frac{1}{5})$&`$+$'&$2/3$&$2/3$&$4/3$\\
&&`$-$'&$5/3$&$5/3$&$1$\\
II&$(0,\frac{2}{5})$&`$+$'&$-2/3$&$2/3$&$4/3$\\
&&`$-$'&$3$&$5/3$&$1$\\
&$(\frac{1}{5},\frac{1}{5})$&`$+$'&$0$&$0$&$4/3$\\
&&`$-$'&$7/3$&$7/3$&$1$\\
&$(\frac{1}{5},\frac{2}{5})$&`$+$'&$-1/3$&$1/3$&$2$\\
&&`$-$'&$8/3$&$2$&$1/3$\\
&$(\frac{3}{10},\frac{3}{10})$&`$+$'&$0$&$0$&$2$\\
&&`$-$'&$7/3$&$7/3$&$1/3$\\
III&$(\frac{1}{5},\frac{2}{5})$&`$+$'&$-1/3$&$1/3$&$2$\\
&&`$-$'&$8/3$&$2$&$1/3$\\
&$(\frac{3}{10},\frac{3}{10})$&`$+$'&$0$&$0$&$2$\\
&&`$-$'&$7/3$&$7/3$&$1/3$\\
&$(\frac{2}{5},\frac{2}{5})$&`$+$'&$1/3$&$1/3$&$8/3$\\
&&`$-$'&$2$&$2$&$-1/3$\\
VI&$(0,\frac{3}{5})$&`$+$'&$0$&$1$&$1$\\% fix from here
&&`$-$'&$0$&$4/3$&$4/3$\\
&$(0,1)$&`$+$'&$0$&$7/3$&$7/3$\\
&&`$-$'&$0$&$0$&$0$\\
&$(\frac{2}{5},\frac{3}{5})$&`$+$'&$0$&$1$&$7/3$\\
&&`$-$'&$0$&$4/3$&$0$\\
\bottomrule
\end{tabular}
\end{table}

\subsubsection{Region IV}
\label{subsubsec:new_IV}
The argument here is analogous with \ref{subsubsec:old_III}: 
$(1-\frac{25}{9}yz)x\leq (1-\frac{25}{9}yz)(1-x)$ as $x\leq \frac{1}{2}$;
$\frac{5}{3}y+\frac{5}{3}z-\frac{50}{9}yz\geq (1-\frac{5}{3}y)(1-\frac{5}{3}z)$ as
\[\frac{5}{3}y+\frac{5}{3}z-\frac{50}{9}yz- \left(1-\frac{5}{3}y\right)\left(1-\frac{5}{3}z\right)=\frac{1-(2-5y)(2-5z)}{3}\]
and $y,z\in \left[2/5,3/5\right]$.
Therefore, $\frac{10}{3}\cdot e.lp_w-e.cost_w$ is smaller with sign `$+$'.

Replacing $x=z-y$,
\[\frac{10}{3}\cdot e.lp_w-e.cost_w=-\frac{250}{27}yz^2+\frac{250}{27}y^2z+\frac{50}{9}yz-5y+\frac{5}{3}z.\]
Fixing $y$, this formula is concave of $z$, hence minimized in either $z=y$ or $z=3/5$, which yield $\frac{50}{9}y^2-\frac{10}{3}y$ and $\frac{50}{9}y^2-5y+1$ each.
Within the range $y\in[2/5,3/5]$, the minimum value of each is $-1/9$ in $y=2/5$ and $-1/8$ in $y=9/20$.
Thus, $\frac{10}{3}\cdot e.lp_w-e.cost_w\geq -1/8$.

Finally, $\frac{10}{3}\cdot e.lp_u-e.cost_y\geq 1/3$, minimized at $(y,z)=(2/5,3/5)$ with `$+$' sign; $\frac{10}{3}\cdot e.lp_v-e.cost_v\geq 2/3$, minimized at $(y,z)=(2/5,2/5)$ with `$+$' sign.

Therefore, any objective value is at least $1/3-1/8=5/24\geq 0$ in the region.

\subsubsection{Region V}
\label{subsubsec:new_V}
The argument here is analogous with \ref{subsubsec:old_IV}:
$\frac{10}{3}\cdot e.lp_w-e.cost_w$ is smaller with sign `$+$', which is $(\frac{10}{3}x-1)(1-\frac{5}{3}y)$.
This has a minimum value of $-1/9$ in $(x,y)=(1/5,2/5)$.

The formula for $\frac{10}{3}\cdot e.lp_v-e.cost_v$ is either $5/3(2x+y)$ or $7/3-5/3(2x+y)$, whose global minimum value is $0$ in $(x,y)=(2/5,3/5)$ with sign `$-$'.

The formula for $\frac{10}{3}\cdot e.lp_u-e.cost_u$ is either $10/3\cdot y -1$ or $10/3(1-y)$, whose global minimum value is $1/3$ in $y=2/5$ with sign `$+$'.
Thus, $(\frac{10}{3}\cdot e.lp_w-e.cost_w)+(\frac{10}{3}\cdot e.lp_u-e.cost_u)$ and $(\frac{10}{3}\cdot e.lp_u-e.cost_u)+(\frac{10}{3}\cdot e.lp_v-e.cost_v)$ are at least $-1/9+1/3=2/9\geq 0$.

Finally,
\begin{align*}
    &\left(\frac{10}{3}\cdot e.lp_w-e.cost_w\right)+\left(\frac{10}{3}\cdot e.lp_v-e.cost_v\right)\\
    &\quad\geq \min\left\{\left(\frac{10}{3}x-1\right)\left(1-\frac{5}{3}y\right)+\frac{5}{3}(2x+y),\left(\frac{10}{3}x-1\right)\left(1-\frac{5}{3}y\right)+\frac{7}{3}-\frac{5}{3}(2x+y)\right\}\\
    &\quad=\min\left\{\left(\frac{10}{3}x-2\right)\left(2-\frac{5}{3}y\right)+3, \frac{50}{9}xy+\frac{4}{3} \right\}\\
    &\quad\geq \min\left\{1,\frac{4}{3}\right\}\geq 1\geq 0.
\end{align*}
\subsubsection{Region VII}
\label{subsubsec:new_VII}
The argument here is analogous with \ref{subsubsec:old_VI}:
$\frac{10}{3}\cdot e.lp_w-e.cost_w$ is smaller with sign `$+$', which is $(\frac{10}{3}x-1)(1-\frac{5}{3}y)$.
This has a minimum value of $0$ in $(x,y)=(0.4,0.6)$.

The formula for $\frac{10}{3}\cdot e.lp_u-e.cost_u$ is either $(1-\frac{5}{3}x)(\frac{10}{3}y-1)$ or $\frac{10}{3}(1-\frac{5}{3}x)(1-y)$, whose global minimum value is $1/9$ in $(x,y)=(2/5,2/5),(1/2,1/2)$ with sign `$+$'.

For $\frac{10}{3}\cdot e.lp_v-e.cost_v$,
\begin{align*}
    &\frac{10}{3}\left(1-\frac{25}{9}xy\right)(2z-1)-\left(
    \frac{5}{3}x+\frac{5}{3}y-\frac{50}{9}xy-\left(1-\frac{5}{3}x\right)\left(1-\frac{5}{3}y\right)\right)\\
    &\quad\geq 2\left(1-\frac{25}{9}xy\right)-\left(-\frac{25}{3}xy+\frac{10}{3}x+\frac{10}{3}y-1\right)\\
    &\quad=\frac{25}{9}xy-\frac{10}{3}x-\frac{10}{3}y+3\\
    &\quad = \left(2-\frac{5}{3}x\right)\left(2-\frac{5}{3}y\right)-1\geq \frac{1}{3}\geq 0,
\end{align*}
since $z\geq\frac{4}{5}$.
Therefore, the formula with sign `$-$' is smaller, which is $\frac{10}{3}(1-\frac{25}{9}xy)(1-z)-(1-\frac{5}{3}x)(1-\frac{5}{3}y)$.

Define $p=xy$, then
\[\frac{10}{3}\left(1-\frac{25}{9}xy\right)(1-z)-\left(1-\frac{5}{3}x\right)\left(1-\frac{5}{3}y\right)=\frac{250}{27}pz-\frac{5}{3}z-\frac{325}{27}p+\frac{7}{3},\]
with region $z\in [4/5,1],\,p\in\left[\frac{2}{5}z-\frac{4}{25},\left(\frac{z}{2}\right)^2\right]$.
Fixing the value of $z$,it is a decreasing function with respect to $p$. Therefore, the function is minimized when $p=\left(\frac{z}{2}\right)^2$, which now becomes:
\[\frac{125}{54}z^3-\frac{325}{108}z^2-\frac{5}{3}z+\frac{7}{3}.\]
This function decreases in the range $z\in[4/5,1]$, hence the minimum value is $-1/36$ in $z=1$.
Therefore, $(\frac{10}{3}\cdot e.lp_w-e.cost_w)+(\frac{10}{3}\cdot e.lp_u-e.cost_u)$ and $(\frac{10}{3}\cdot e.lp_u-e.cost_u)+(\frac{10}{3}\cdot e.lp_v-e.cost_v)$ are at least $1/9-1/36=1/12\geq 0$.

Finally,
\begin{align*}
    &\left(\frac{10}{3}\cdot e.lp_w-e.cost_w\right)+\left(\frac{10}{3}\cdot e.lp_v-e.cost_v\right)\\
    &\quad\geq\left(\frac{10}{3}x-1\right)\left(1-\frac{5}{3}y\right)+\frac{10}{3}\left(1-\frac{25}{9}xy\right)(1-z)-\left(1-\frac{5}{3}x\right)\left(1-\frac{5}{3}y\right)\\
    &\quad = (5x-2)\left(1-\frac{5}{3}y\right)+\frac{10}{3}\left(1-\frac{25}{9}xy\right)(1-z)\\
    &\quad \geq 0\cdot 0+\frac{10}{3}\cdot \frac{11}{36}\cdot 0 =0,
\end{align*}
since $x\geq \frac{2}{5},\,y\leq\frac{3}{5},\,z\leq 1$.

\subsection{ All coordinates belong to endpoint of segments.}
There are also $4$ tuples of $(x,y,z)$ to verify, which are fully covered in Table~\ref{tab:wCC_bd_new}. Note that there are now only single point $2/5$ where rounding functions are not continuous.
\begin{table}[]
\caption{Value of $\frac{10}{3}\cdot e.lp-e.cost$ for each points in consideration of Case 2.}
    \label{tab:wCC_bd_new}
    \centering
\begin{tabular}{llrrr}
\toprule
$(x,y,z)$&Sign& $\frac{10}{3}\cdot e.lp_w-e.cost_w$& $\frac{10}{3}\cdot e.lp_u-e.cost_u$&$\frac{10}{3}\cdot e.lp_v-e.cost_v$\\
\midrule % change from here
$(\frac{3}{5},\frac{3}{5},1)$&`$+$'&$0$&$0$&$0$\\
&`$-$'&$0$&$0$&$0$\\
$(\frac{2}{5}-\delta,1,1)$&`$+$'&$0$&$7/3$&$7/3$\\
&`$-$'&$0$&$0$&$0$\\
$(\frac{2}{5}+\delta,1,1)$&`$+$'&$0$&$7/9$&$7/9$\\
&`$-$'&$0$&$0$&$0$\\
$(\frac{3}{5},1,1)$&`$+$'&$0$&$0$&$0$\\
&`$-$'&$0$&$0$&$0$\\
$(1,1,1)$&`$+$'&$0$&$0$&$0$\\
&`$-$'&$0$&$0$&$0$\\
\bottomrule
\end{tabular}
\end{table}
\\
\\
This completes our case analysis.

\end{document}